\documentclass[manuscript,screen,nonacm]{acmart}

\usepackage{amsmath,amsthm}
\usepackage{algorithm}
\usepackage{algpseudocode}
\usepackage{graphicx}
\usepackage{booktabs}
\usepackage{multirow}
\usepackage{xcolor}
\usepackage{url}

\usepackage{pifont}
\usepackage{hyperref}
\usepackage{cleveref}

\newtheorem{theorem}{Theorem}

\usepackage{array}
\usepackage{tabularx}

\usepackage{pgfplots}
\pgfplotsset{compat=1.18}
\usetikzlibrary{arrows.meta,positioning, calc}

\newcommand{\xmark}{$\times$}
\AtBeginDocument{%
  }

\setcopyright{none}
\renewcommand\footnotetextcopyrightpermission[1]{}

\begin{document}

\title[SAEID: Secure Aggregate Encryption for Multi-Vendor FPGA Cloud Deployment]{Secure Aggregate Encryption with Identity-Based Authentication for Multi-Vendor FPGA Cloud Deployment}

\author{Mukta Debnath}
\email{muktapearl.debnath@gmail.com}
\orcid{0000-0001-5478-5163}
\affiliation{%
  \institution{University of Calcutta}
  \city{Kolkata}
  \country{India}}
  
\author{Krishnendu Guha}
\email{KGuha@ucc.ie}
\affiliation{%
  \institution{University College Cork}
  \state{Cork}
  \country{Ireland}
}

\author{Debasri Saha}
\email{debasri\_cu@yahoo.in}
\affiliation{%
  \institution{University of Calcutta}
  \city{Kolkata}
  \country{India}
}

\author{Amlan Chakrabarti}
\email{acakcs@caluniv.ac.in}
\affiliation{%
  \institution{University of Calcutta}
  \city{Kolkata}
  \country{India}
}

\author{Susmita Sur-Kolay}
\email{ssk@isical.ac.in}
\affiliation{%
  \institution{Indian Statistical Institute}
  \city{Kolkata}
  \country{India}}

\renewcommand{\shortauthors}{Debnath et al.}
\pagenumbering{arabic}

\begin{abstract}

Secure deployment of FPGA bitstreams in heterogeneous multi-vendor cloud
infrastructures requires scalable authorization, authenticated device
access, and bitstream protection throughout the deployment lifecycle.
Existing approaches typically address these functions through separate
mechanisms, increasing coordination and key-management requirements.
This paper presents \emph{SAEID}, a secure FPGA deployment framework that
integrates aggregate authorization, certificate-free identity-based device
authentication, and identity-bound bitstream verification within a
pairing-based cryptographic framework, with symmetric cryptography used for
session binding and AES-256-GCM-based bitstream protection.
Building on AgEID, an aggregate-encryption scheme that enables individual
decryption for authorized FPGA devices, SAEID extends the underlying
framework to heterogeneous multi-vendor deployments by supporting multiple
FPGA vendors and IP providers, capability-aware authorization, and dynamic
device membership. SAEID provides protection against unauthorized access
to future deployments after device revocation and to prior deployments by
newly enrolled devices, while retaining constant-size aggregate ciphertexts
within each vendor domain and individual decryption.
Experimental results demonstrate the practical performance of SAEID, with
identity-based device authentication completing in approximately
4.35~ms and device membership updates requiring 4.95--5.09~$\mu$s for
device addition. The aggregate-encryption component exhibits scalable
behavior with increasing device-set size.
The complete SAEID software decryption path was also validated on a
physical ZC702 Cortex-A9 platform, requiring 191.126~ms.
These results demonstrate scalable aggregate authorization with bounded
authentication and dynamic-membership overhead for secure multi-vendor
FPGA deployment.

\end{abstract}




\maketitle

\section{Introduction}

Field-programmable gate arrays (FPGAs) have become shared cloud resources
through the FPGA-as-a-Service (FaaS) paradigm, enabling remote deployment
of accelerator IPs on heterogeneous FPGA infrastructures~\cite{jin2020survey,zeitouni2020multitenant}.
Unlike conventional embedded deployments, cloud-based FPGA environments
involve shared infrastructure managed by cloud service providers (CSPs),
heterogeneous FPGA resources, and multiple tenants or IP providers that
deploy proprietary hardware designs~\cite{zeitouni2021trust,
turan2021trust,ahmed2025multitenant}. This deployment model introduces
security requirements beyond conventional bitstream confidentiality,
including device authentication, device-bound bitstream verification,
scalable authorization, and secure management of dynamically changing
device groups.

Existing FPGA protection mechanisms, including secure boot and encrypted
configuration, primarily protect individual devices
~\cite{trimberger2011authenticated,trimberger2014features,amd_ug1085_zynqmp}.
Certificate-based public-key infrastructures (PKIs) provide device
authentication through digital certificates, but require certificate
management, including storage, validation, and revocation, which can
introduce additional communication, latency, and storage overheads
\cite{singh2019hybrid,hu2019certificate}.
Encryption of FPGA bitfiles is commonly used to prevent unauthorized
access to protected IP, while authentication mechanisms are required to
establish the legitimacy of IP vendors and users
\cite{saha2017fpga}.
However, these mechanisms do not by themselves provide scalable
authorization and authenticated deployment across heterogeneous
multi-vendor FPGA infrastructures.

Aggregate encryption provides scalable authorization by allowing a single
ciphertext to distribute a session key to multiple authorized recipients
while preserving individual decryption~\cite{bgw2005,kac2018}. Our previous
work, AgEID~\cite{debnath2024ageid}, applied this principle to FPGA
bitstream deployment, enabling an IP provider to encrypt a session key once
for an authorized FPGA cluster. However, AgEID assumes a homogeneous
single-vendor environment and primarily addresses confidentiality. Device
authentication, bitstream integrity and device binding, heterogeneous
multi-vendor deployment, and dynamic device membership are outside its
security model.

This paper presents \textbf{SAEID}, a secure FPGA deployment framework that
extends AgEID to heterogeneous multi-vendor cloud infrastructures. SAEID
supports multiple FPGA vendors and IP providers, capability-aware
authorization, dynamic device membership, certificate-free identity-based
device authentication, and identity-bound bitstream verification while
retaining constant-size aggregate ciphertexts within each vendor domain
and individual decryption. The framework combines aggregate
authorization with authenticated deployment sessions and device-bound
bitstream protection while maintaining separate cryptographic domains for
the corresponding security functions. 

\begin{table}[t]
\centering
\caption{Comparison Between AgEID and SAEID}
\label{tab:ageid_vs_saeid}
\begin{tabular}{lcc}
\hline
\textbf{Capability} & \textbf{AgEID} & \textbf{SAEID} \\
\hline
Aggregate Encryption & $\checkmark$ & $\checkmark$ \\
Individual Decryption & $\checkmark$ & $\checkmark$ \\
Multi-Vendor FPGA Deployment & $\times$ & $\checkmark$ \\
Multiple IP Providers & $\times$ & $\checkmark$ \\
Dynamic Device Membership & Partial & $\checkmark$ \\
Capability-Aware Authorization & $\times$ & $\checkmark$ \\
Forward/Backward Security & $\times$ & $\checkmark$ \\
Certificate-Free Device Authentication & $\times$ & $\checkmark$ \\
Identity-Bound Bitstream Verification & $\times$ & $\checkmark$ \\
Authenticated Provider Binding & $\times$ & $\checkmark$ \\
\hline
\end{tabular}
\end{table}

The contributions of this paper are summarized as follows:

\begin{itemize}

\item \textbf{Multi-vendor and multi-provider aggregate authorization.}
We extend aggregate encryption to heterogeneous FPGA clouds using
vendor-specific authorization domains, capability-aware device selection,
and constant-size aggregate ciphertexts within each vendor domain.

\item \textbf{Dynamic membership with forward and backward security.}
We introduce device addition and revocation through constant-cost updates
to the affected vendor-specific authorization component, together with
forward and backward security across membership epochs.

\item \textbf{Certificate-free identity-based device authentication and
device-bound bitstream integrity.}
We integrate pairing-based device authentication and identity-bound
bitstream verification with the aggregate-encryption framework using a
common bilinear-group foundation. The identity domain uses an independent
vendor master secret $s_v$, while aggregate authorization uses
$\gamma_v$, allowing the two functions to share the pairing infrastructure
without sharing cryptographic secrets.

\item \textbf{Security analysis and implementation evaluation.}
We provide security arguments covering confidentiality, device
impersonation and replay resistance, bitstream integrity and device
binding, dynamic membership, and multi-vendor isolation. Experimental
evaluation shows that the identity-based authentication exchange
completes in approximately 4.35~ms, while device addition and revocation
require only 4.95--5.09~$\mu$s and 4.49--5.46~$\mu$s, respectively. The
complete SAEID software path is additionally validated on physical ZC702
ARM hardware.

\end{itemize}

The remainder of this paper is organized as follows. Section~\ref{sec:related}
reviews related work and cryptographic background. Section~\ref{sec:system-threat-model}
defines the system and threat models. Section~\ref{sec:framework_overview}
presents the SAEID framework, followed by security analysis in
Section~\ref{sec:security}, scalability analysis in Section~\ref{sec:scalability},
and experimental evaluation in Section~\ref{sec:eval}.
Section~\ref{sec:discussion} discusses the results and limitations, and
Section~\ref{sec:conclusion} concludes the paper. Detailed correctness
derivations and security proof sketches are provided in
Appendix~\ref{app:proofs}.

\begin{table*}[t]
\centering
\caption{Comparison of Representative FPGA Bitstream Deployment Frameworks}
\label{tab:related}
\resizebox{\textwidth}{!}{%
\begin{tabular}{lcccc}
\toprule
\textbf{Security Capability} &
\textbf{Conventional FPGA Protection} &
\textbf{Broadcast/KAC} &
\textbf{AgEID} &
\textbf{SAEID} \\
\midrule
Constant-Size Ciphertext & \xmark & \xmark & \checkmark & \checkmark \\
Individual Decryption & \checkmark & \checkmark & \checkmark & \checkmark \\
External Trusted Third Party & Often Required & Often Required & \xmark & \xmark \\
Multi-Vendor Deployment & \xmark & Limited & Limited & \checkmark \\
Multiple IP Providers & \xmark & \xmark & \xmark & \checkmark \\
Dynamic Device Membership & \xmark & Partial & Partial & \checkmark \\
Forward/Backward Security & \xmark & Partial & Not Formalized & \checkmark \\
Certificate-Free Device Authentication & \xmark & \xmark & \xmark & \checkmark \\
Identity-Bound Bitstream Verification & Vendor-Specific & \xmark & \xmark & \checkmark \\
\bottomrule
\end{tabular}}
\end{table*}

\section{Background and Related Work}
\label{sec:related}

\subsection{FPGA Bitstream Security and Cloud Deployment}

FPGA bitstream security has traditionally focused on protecting the
confidentiality and integrity of configuration data at the device
level. Commercial FPGA platforms provide hardware-supported mechanisms
for secure configuration and authentication~\cite{amd_ug1085_zynqmp,
trimberger2014features,trimberger2011authenticated}. Prior research
has further investigated bitstream authentication~\cite{drimer2007authentication},
tamper protection~\cite{hoque2019hidden}, fault-injection attacks
~\cite{swierczynski2018bifi}, and secure remote configuration
~\cite{drimer2009secure}. FPGA IP protection has additionally considered
confidentiality, authentication, tamper detection, and protection against
hardware Trojan insertion. Saha and Sur-Kolay~\cite{saha2017fpga}
discuss security challenges in FPGA IP distribution and management,
encrypted bitstream deployment, IP vendor and user authentication, and
partially reconfigurable FPGA systems. These mechanisms provide
important device- and IP-level security guarantees but primarily
consider individual FPGA configurations or IP distribution scenarios.

With the emergence of cloud and multi-tenant FPGA infrastructures,
security requirements extend to shared resources and remote
provisioning. Studies have identified threats involving unauthorized
configuration, information leakage, and attacks arising from
multi-tenant resource sharing~\cite{jin2020survey,zeitouni2020multitenant,
turan2021trust}. Secure cloud-FPGA systems have consequently explored
protected execution and bitstream protection~\cite{hong2018fasten,
kim2019safedb}, while trusted configuration mechanisms have been
proposed for cloud FPGA deployments~\cite{zeitouni2021trust}.
Bag \textit{et al.}~\cite{bag2020} further investigated
cryptographically secure provisioning for multi-tenant FPGA
environments.
FPGA intellectual-property protection has also been studied through
dynamic IP protection and licensing mechanisms
~\cite{guneysu2007dynamic,maes2012payperuse}, as well as protection
across the design-to-bitstream flow~\cite{hutchings2024fpga}. More
recently, multi-tenant cloud FPGA security has been surveyed with
emphasis on security, trust, and privacy challenges
~\cite{ahmed2025multitenant}. 

\subsection{Aggregate Encryption and Dynamic Authorization}

Aggregate encryption addresses the scalability challenges of secure FPGA
deployment by enabling cryptographic material to be distributed to
multiple authorized devices without requiring a separate encryption
operation for each recipient. Broadcast encryption provides a foundation
for such multi-recipient authorization~\cite{bgw2005}, while
key-aggregate cryptography (KAC) enables a compact aggregate key to
provide access to multiple ciphertext classes~\cite{kac2018}. Building
on these concepts, AgEID~\cite{debnath2024ageid} adapted aggregate
encryption to FPGA bitstream deployment, allowing an IP provider to
encrypt an AES key once for an authorized FPGA cluster while
preserving individual device decryption. Bag \textit{et al.}~\cite{bag2020}
also demonstrated the applicability of cryptographic aggregation to
secure multi-tenant FPGA provisioning.
Cloud FPGA environments additionally require authorization policies to
evolve as devices are added or revoked. Dynamic key-aggregate
cryptosystems~\cite{patranabis2015dynamic,patranabis2017provably}
address dynamic membership and secure key updates, while secure FPGA
provisioning and configuration have been explored using cryptographic
key-establishment and re-encryption techniques
~\cite{drimer2009secure,eguro2012fpgas,turan2020proxy}. 

\subsection{Device Authentication and Bitstream Authenticity}
\label{sec:related_auth}

While aggregate encryption provides confidentiality and scalable
authorization, secure FPGA deployment also requires authentication of the
target device and verification of the deployed bitstream. Device
authentication establishes the legitimacy of the FPGA participating in a
deployment, whereas bitstream authenticity ensures that the configuration
originates from an authorized IP provider and has not been modified.
FPGA security architectures therefore incorporate authentication and
integrity mechanisms in addition to bitstream encryption
~\cite{trimberger2011authenticated,trimberger2014features,
amd_ug1085_zynqmp}. Earlier work by Drimer
~\cite{drimer2007authentication} investigated authentication of encrypted
FPGA bitstreams using MAC-based mechanisms. These approaches provide
authentication and integrity within their respective deployment models,
but remain separate from the aggregate-encryption mechanism used for
scalable multi-device authorization.

Hardware-rooted authentication has also been extensively investigated for
FPGA devices using physical unclonable functions (PUFs), including
time-bounded FPGA authentication and PUF-based authentication and
key-exchange protocols
~\cite{majzoobi2011timebounded,rostami2014pufauth}. 
Identity-based
cryptography provides an alternative certificate-free approach in which
public keys are derived from identities~\cite{shamir1984}. 
Pairing-based
identity-based cryptography subsequently enabled identity-based encryption
and signature constructions~\cite{boneh2001ibe,chacheon2003,
hess2002efficient,sakai2003}. 
FPGA implementations of identity-based
encryption and identity-based signature verification have also been
demonstrated using pairing-based and elliptic-curve techniques
~\cite{kerins2004ibe,jarvinen2007selfcertified}.
Identity-based signcryption further combines confidentiality and
authenticity within pairing-based frameworks
~\cite{malonelee2002signcryption,libert2003signcryption,
boyen2003signcryption}, while Chatterjee \textit{et al.}
~\cite{chatterjee2017pufibe} explored pairing-based identity mechanisms
for hardware authentication.

\subsection{Research Gap}

Existing FPGA security mechanisms address important aspects of bitstream
protection, scalable authorization, and device authentication, but do not
jointly address the requirements of heterogeneous cloud deployment.
Vendor-supported secure-boot mechanisms primarily provide
device-specific bitstream protection, while aggregate-encryption
frameworks such as AgEID~\cite{debnath2024ageid} reduce the overhead of
authorizing multiple devices without providing certificate-free device
authentication or identity-bound bitstream verification. Conventional
authentication mechanisms may additionally require certificate and key
management infrastructure. Consequently, there remains a need for a
framework that supports scalable aggregate authorization,
Certificate-Free Identity-Based Device Authentication, identity-bound bitstream
verification, heterogeneous multi-vendor deployment, multiple IP
providers, and dynamic device membership. 
Table~\ref{tab:related} summarizes the positioning of SAEID with respect
to FPGA bitstream deployment approaches.

\begin{figure}[t]
\centering
\begin{tikzpicture}[
node distance=6mm and 8mm,
box/.style={
    draw,
    rounded corners,
    align=center,
    minimum height=0.85cm,
    font=\scriptsize
},
arr/.style={-{Latex[length=2mm]}, thick}
]

\node[box, fill=purple!10, minimum width=2.3cm] (ia1)
    {Identity Authority\\$\mathcal{IA}_{v_1}$: $s_{v_1}$};

\node[box, fill=orange!12, minimum width=2.3cm, below=of ia1] (fv1)
    {FPGA Vendor $v_1$\\Setup/KeyGen};

\node[box, fill=purple!10, minimum width=2.3cm, below=1.4cm of fv1] (ia2)
    {Identity Authority\\$\mathcal{IA}_{v_2}$: $s_{v_2}$};

\node[box, fill=orange!12, minimum width=2.3cm, below=of ia2] (fv2)
    {FPGA Vendor $v_2$\\Setup/KeyGen};

\node[box, fill=gray!12, minimum width=2.6cm,
      right=1.7cm of fv1, yshift=-1.1cm] (csp)
    {partially trusted CSP\\Cluster formation,\\message relay};

\node[box, fill=blue!8, minimum width=2.5cm,
      right=1.6cm of csp] (ipp)
    {IP Provider $p$\\Authentication /\\Bitstream Release};

\node[box, fill=green!10, minimum width=1.7cm,
      above right=0.1cm and 1.2cm of csp] (d1)
    {FPGA$_{v_1,1}$\\Authenticated};

\node[box, fill=green!10, minimum width=1.7cm,
      below right=0.1cm and 1.2cm of csp] (d2)
    {FPGA$_{v_2,1}$\\Authenticated};

\draw[arr] (ia1) -- (fv1)
    node[midway,right,font=\tiny]{$s_{v_1}$};

\draw[arr] (ia2) -- (fv2)
    node[midway,right,font=\tiny]{$s_{v_2}$};

\draw[arr] (fv1) -- (d1)
    node[midway,above,font=\tiny]
    {$d_{v,i},\mathrm{SK}_{\mathrm{ID}_{v,i}}$};

\draw[arr] (fv2) -- (d2)
    node[midway,below,font=\tiny]
    {$d_{v,i},\mathrm{SK}_{\mathrm{ID}_{v,i}}$};

\draw[arr] (fv1) -- (csp)
    node[midway,above,font=\tiny]
    {$PK_v,P_{\mathrm{pub},v}$};

\draw[arr] (fv2) -- (csp)
    node[midway,below,font=\tiny]
    {$PK_v,P_{\mathrm{pub},v}$};

\draw[arr] (ia1.east) to[bend left=18]
    node[midway,above,font=\tiny]
    {$\mathrm{SK}_{\mathrm{ID}_{p,v_1}}$} (ipp.north west);

\draw[arr] (ia2.east) to[bend right=18]
    node[midway,below,font=\tiny]
    {$\mathrm{SK}_{\mathrm{ID}_{p,v_2}}$} (ipp.south west);

\draw[arr] (csp) -- (ipp)
    node[midway,above,font=\tiny]{relay};

\draw[arr] (d1) -- (csp)
    node[midway,above,font=\tiny]{Auth. proof};

\draw[arr] (d2) -- (csp)
    node[midway,below,font=\tiny]{Auth. proof};

\draw[arr] (ipp) -- (csp)
    node[midway,below,font=\tiny]{Challenge, ciphertext};

\end{tikzpicture}%
\caption{SAEID system architecture. Each vendor maintains a separate
cryptographic domain and provisions its FPGA devices with individual
decryption keys and vendor-scoped identity keys. Registered IP providers
obtain corresponding identity credentials. The CSP manages
resources and relays protocol messages without access to vendor, provider,
or device secret keys.}
\Description{SAEID system architecture diagram.}
\label{fig:architecture}
\end{figure}
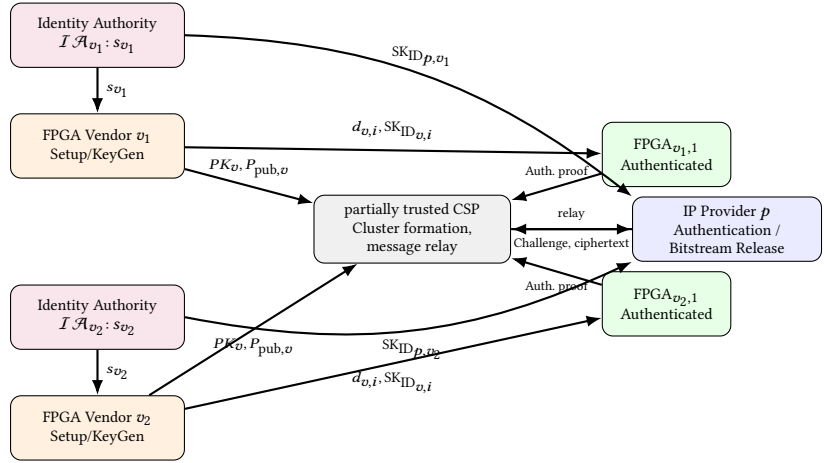

\section{System Model and Threat Model}
\label{sec:system-threat-model}

SAEID considers a cloud-based FPGA deployment environment involving FPGA
vendors, IP providers, a cloud service provider (CSP), and FPGA devices,
consistent with heterogeneous and multi-tenant FPGA-cloud deployments
\cite{zeitouni2021trust,turan2021trust,ahmed2025multitenant,
skhiri2019fpga}. Each FPGA vendor operates an identity authority
$\mathcal{IA}_v$ for its own device and provider registrations. SAEID
further considers multiple FPGA vendors and IP providers, with each vendor
maintaining a separate cryptographic domain.

\subsection{System Entities}

As shown in~\autoref{fig:architecture}, SAEID comprises five classes of
entities, reflecting the separation of hardware ownership, IP ownership,
and cloud infrastructure commonly found in FPGA-cloud deployment
models~\cite{zeitouni2021trust,turan2021trust}.

\begin{itemize}

\item \textbf{Identity Authorities ($\mathcal{IA}_v$).}
Each FPGA vendor $v\in\mathcal{V}$ operates an identity authority
$\mathcal{IA}_v$ for registering its devices and IP providers. Thus,
$\mathcal{IA}_v$ is a vendor-operated identity-provisioning function
rather than an independent entity. It maintains the vendor-specific
master identity secret $s_v$ and public parameter $P_{\mathrm{pub},v}$,
and derives identity-based private keys
$\mathrm{SK}_{\mathrm{ID}_{v,i}}$ for devices and
$\mathrm{SK}_{\mathrm{ID}_{p}^{(v)}}$ for registered providers.

\item \textbf{FPGA Vendors ($\mathcal{V}$).}
Vendor $v$ manufactures and provisions its FPGA devices and maintains its
vendor-specific cryptographic domain. Each device
$\mathrm{FPGA}_{v,i}$ is provisioned with its individual decryption key
$d_{v,i}$ and identity private key
$\mathrm{SK}_{\mathrm{ID}_{v,i}}$. The vendor publishes the corresponding
aggregate-encryption and identity public parameters.

\item \textbf{IP Providers ($\mathcal{P}$).}
Provider $p$ develops accelerator designs and prepares
architecture-specific bitstream variants. For deployment in vendor domain
$v$, it obtains the vendor-scoped identity credential
$\mathrm{SK}_{\mathrm{ID}_{p}^{(v)}}$, selects eligible devices, generates
a fresh 256-bit AES key $m$, and invokes SAEID to authorize its recovery by
the selected devices.

\item \textbf{Cloud Service Provider (CSP).}
The CSP manages the shared FPGA infrastructure, including device inventory,
resource allocation, scheduling, storage, and message forwarding. It is
partially trusted: it is authorized to access public system parameters and
to perform resource-management and message-relay functions, but does not
receive vendor master secrets, identity-authority secrets, provider private
credentials, device private keys, session keys, or deployment AES keys. The
CSP may store, reorder, delay, or relay deployment messages and ciphertexts.

\item \textbf{FPGA Devices.}
Each $\mathrm{FPGA}_{v,i}$ has a unique identity
$\mathrm{ID}_{v,i}$, individual decryption key $d_{v,i}$, and identity key
$\mathrm{SK}_{\mathrm{ID}_{v,i}}$. The identity key supports device
authentication, while $d_{v,i}$ enables recovery of the authorized AES key
$m$. The device accepts a bitstream only after the required authentication,
authorization, and integrity checks succeed.

\end{itemize}

\subsection{Capability-Aware Deployment Model}
\label{sec:capability_model}

Capability-based access control has been investigated for distributed and
cloud environments to provide scalable and fine-grained resource
authorization~\cite{anggorojati2018capability}. In FPGA-cloud
deployments, resource and device selection must additionally account for
heterogeneous FPGA capabilities and available resources
\cite{skhiri2019fpga,cong2016heterogeneous}. SAEID therefore selects
FPGA devices according to their deployment capabilities rather than
assuming a homogeneous device pool.
Let
\begin{equation}
S_c =
\left\{
\mathrm{FPGA}_{v,i}
\mid
\mathrm{Capability}(\mathrm{FPGA}_{v,i})=c
\right\},
\label{eq:capability_cluster}
\end{equation}
denote the devices satisfying capability class $c$, where a capability
class captures requirements such as FPGA-family compatibility, supported
accelerator features, available logic, memory, DSP resources, and bitstream
compatibility.

For vendor $v$, the eligible subset is
\begin{equation}
S_{c,v}
=
\left\{
\mathrm{FPGA}_{v,i}\in S_c
\right\},
\label{eq:vendor_capability_subset}
\end{equation}
and the participating vendor set is
\begin{equation}
\mathcal{V}_{S_c}
=
\left\{
v\in\mathcal{V}
\mid
S_{c,v}\neq\emptyset
\right\}.
\label{eq:participating_vendors}
\end{equation}

A capability cluster may therefore span multiple vendor domains. Capability
membership determines eligibility only; cryptographic authorization
remains vendor- and provider-specific.

\subsection{Multi-Vendor and Multi-Provider Deployment}
\label{sec:multi_vendor_model}

A selected capability cluster may contain devices from multiple vendors,
reflecting heterogeneous FPGA-cloud infrastructures
\cite{turan2021trust,ahmed2025multitenant}. For each participating
vendor $v$, SAEID maintains a separate provisioning and cryptographic
context and constructs the corresponding vendor-specific authorization
component. Thus, cross-vendor resource pooling does not require sharing
device private keys or merging vendor cryptographic domains.

The system model also allows multiple IP providers to deploy independent
applications over the shared FPGA infrastructure. Each provider
deployment is associated with its own authorization state, session
context, and protected bitstream payload, while vendor-issued device
credentials remain unchanged. Consequently, authorization or session
state
associated with one deployment does not by itself grant access to another
provider's protected deployment.

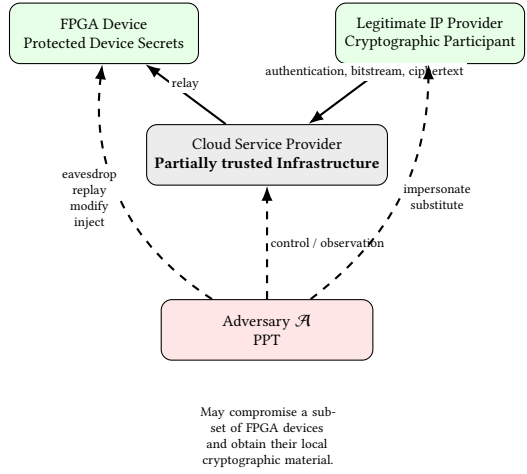
\begin{figure}[t]
\centering
\begin{tikzpicture}[
    node distance=7mm and 10mm,
    box/.style={
        draw, rounded corners, align=center,
        minimum height=0.8cm,
        minimum width=2.5cm,
        font=\scriptsize
    },
    participant/.style={box, fill=green!10},
    untrusted/.style={box, fill=gray!15},
    attacker/.style={
        box,
        fill=red!10,
        minimum width=2.8cm
    },
    arr/.style={
        -{Latex[length=2mm]},
        thick
    },
    attack/.style={
        -{Latex[length=2mm]},
        dashed,
        thick
    },
    label/.style={
        font=\tiny,
        align=center,
        fill=white,
        inner sep=1pt
    }
]

\node[participant] (device)
    {FPGA Device\\Protected Device Secrets};

\node[participant, right=18mm of device] (provider)
    {Legitimate IP Provider\\Cryptographic Participant};

\node[untrusted, below=12mm of $(device)!0.5!(provider)$] (csp)
    {Cloud Service Provider\\\textbf{Partially trusted Infrastructure}};

\draw[arr]
    (provider) --
    node[label, above, pos=0.3]
    {authentication,\;bitstream,\;ciphertext}
    (csp);

\draw[arr]
    (csp) --
    node[label, above]
    {relay}
    (device);

\node[attacker, below=15mm of csp] (adv)
    {Adversary $\mathcal{A}$\\PPT};

\draw[attack]
    (adv) --
    node[label, right]
    {control / observation}
    (csp);

\draw[attack]
    (adv) to[bend left=32]
    node[label, left]
    {eavesdrop\\replay\\modify\\inject}
    (device);

\draw[attack]
    (adv) to[bend right=26]
    node[label, right]
    {impersonate\\substitute}
    (provider);

\node[
    font=\tiny,
    align=center,
    text width=3.0cm,
    below=5mm of adv
]
{May compromise a subset of FPGA devices\\
and obtain their local cryptographic material.};

\end{tikzpicture}

\caption{Threat model of SAEID. The CSP is partially trusted for resource
management and message relay but has no access to vendor, identity-authority,
provider, or device private secrets. The adversary may observe or manipulate
deployment communications, replay or modify messages, attempt device
impersonation or bitstream substitution, and compromise a subset of FPGA
devices, while master secrets of uncompromised domains remain protected.}
\label{fig:threat_model}
\end{figure}

\subsection{Threat Model and Security Objectives}
\label{sec:threat}

SAEID considers a probabilistic polynomial-time (PPT) adversary
$\mathcal{A}$ that may observe or manipulate communications among IP
providers, the CSP, and FPGA devices, illustrated in~\autoref{fig:threat_model},. The adversary may intercept, replay,
modify, reorder, delay, or inject deployment messages and may obtain
public parameters, deployment metadata, and stored ciphertexts.

The adversary may observe or manipulate deployment communications and
attempt unauthorized device access, impersonation, replay, message
modification, or bitstream substitution without possessing valid device
credentials. The adversary cannot obtain the master encryption secret
$\gamma_v$ or identity secret $s_v$ of an honest vendor, or the private
credentials of legitimately registered devices and IP providers. Each
vendor-operated identity authority $\mathcal{IA}_v$ is assumed to securely
protect $s_v$ and provision identity keys only to legitimately registered
devices and IP providers. The CSP is partially trusted for resource
management and message relay but does not possess vendor master secrets,
private credentials, or device private keys.

The model addresses confidentiality, authorization, authentication,
integrity, replay resistance, and related security properties. Availability
and denial-of-service attacks based on message dropping, persistent delay,
or resource exhaustion are outside the scope.

Under this model, SAEID provides the following security objectives:

\begin{itemize}

\item \textbf{Confidentiality and Unauthorized-Access Resistance.}
Only devices authorized for a deployment can recover its AES key $m$ and
decrypt the corresponding bitstream. The CSP and other unauthorized
entities cannot recover $m$ from the authorization metadata and aggregate
ciphertext alone.

\item \textbf{Dynamic-Membership Security.}
A revoked device cannot decrypt ciphertexts generated after revocation,
while a newly enrolled device cannot decrypt ciphertexts generated before
enrollment. After initial provisioning, membership changes affect only the
corresponding vendor-specific aggregate-authorization component and do not
require re-provisioning unaffected devices. Provider credential revocation
is outside the present dynamic-membership model; provider credentials are
assumed valid for their registration lifetime unless the corresponding
vendor identity domain is re-provisioned.

\item \textbf{Domain Isolation.}
Vendor cryptographic domains and provider deployment contexts remain
independent; secret material obtained from one domain does not by itself
provide authorization for another independent domain.

\item \textbf{Authentication and Replay Resistance.}
An adversary without the identity private key of
$\mathrm{FPGA}_{v,i}$ cannot produce an accepted authentication proof for
that device except with negligible probability. Fresh authentication
challenges prevent reuse of previously observed authentication transcripts
in a new deployment session.

\item \textbf{Authenticated Session Binding.}
The authenticated provider--device exchange binds the aggregate ciphertext
to its corresponding deployment context, including the authenticated
identities and freshness values. A ciphertext or authorization message
from a different deployment context must be detected before recovery of
the AES key.

\item \textbf{Bitstream Integrity and Device Binding.}
The identity-bound provider signature authenticates the registered provider
and binds the deployment to the intended device, while AES-256-GCM provides
authenticated protection of the bitstream payload. Consequently,
unauthorized or modified deployment artifacts can be detected before FPGA
configuration. Provider revocation remains outside the present
dynamic-membership mechanism.

\end{itemize}

Physical extraction, invasive attacks, and compromise of an honest
vendor's or identity authority's master secrets are outside the model.
Likewise, compromise of an honest provider's protected credential is
excluded. Therefore, the authentication guarantee provides cryptographic
impersonation resistance but does not cover physical extraction and
duplication of a device's identity key.

\section{Proposed SAEID Framework}
\label{sec:framework_overview}

SAEID extends AgEID from a homogeneous, single-vendor FPGA deployment
model to a heterogeneous multi-vendor and multi-provider environment.
The underlying aggregate-encryption mechanism and individual-decryption
property of AgEID are retained. SAEID adds four principal capabilities:
(i) capability-aware device selection across heterogeneous FPGA platforms;
(ii) vendor-specific authorization domains supporting multiple IP providers;
(iii) dynamic device membership with forward and backward security; and
(iv) certificate-free identity-based device authentication and
identity-bound bitstream verification.

The distinction between AgEID and SAEID is therefore primarily at the
deployment and security layers rather than in the underlying
aggregate-encryption primitive. In SAEID, eligible devices are first
organized into vendor-specific authorization domains, while
identity-based authentication establishes the identity of the target
device before deployment. A fresh AES deployment key $m$ is protected through
the corresponding vendor-specific aggregate-encryption components, and
the architecture-specific bitstream is cryptographically bound to the
authenticated device and provider. Device addition and revocation update
only the affected vendor-specific authorization component, while the
identity credentials of unaffected devices remain unchanged. Table~\ref{tab:ageid_vs_saeid} summarizes
the principal extensions introduced by SAEID.

The following subsections describe the SAEID construction, system
parameters, algorithms, authentication protocol, aggregate authorization,
dynamic membership, and bitstream verification mechanisms. The dynamic membership updates are illustrated in~\autoref{fig:membership},
the deployment workflow in Fig.~\ref{fig:enc-workflow}, and the
authentication protocol in Fig.~\ref{fig:auth-protocol}.

\subsection{From AgEID to SAEID: A Deployment Scenario}
\label{subsec:ageid-saeid-example}

To illustrate the evolution from the original single-vendor deployment
model to the complete SAEID framework, consider a deployment involving
four FPGA devices from vendor $v_1$. The same deployment is then extended
to a second vendor, followed by device membership change, authenticated
deployment, and device-bound bitstream verification.

\paragraph{Single-vendor deployment.}
Vendor $v_1$ operates $n_{v_1}=4$ provisioned FPGA devices,
$\{\mathrm{FPGA}_{v_1,1},\ldots,\mathrm{FPGA}_{v_1,4}\}$, under an
independent aggregate-encryption domain. The vendor maintains the public
group elements required by the construction and a master secret
$\gamma_{v_1}$, and each device $i$ receives an individual decryption key
$d_{v_1,i}=\gamma_{v_1}g_{v_1,i}$.

Suppose that a deployment is intended for the authorized subset
\[
S_{c,v_1}=\{1,2,4\},
\]
while $\mathrm{FPGA}_{v_1,3}$ remains outside the selected deployment
cluster. The corresponding aggregate-authorization component is

\begin{equation}
K_{S_{c,v_1},v_1}
=
g_{v_1,4}+g_{v_1,3}+g_{v_1,1}.
\end{equation}

The provider generates a fresh group-valued deployment secret
$m_G\in G_T$ and derives the AES-256 deployment key as

\begin{equation}
K_{\mathrm{AES}}
=
\mathrm{SHA256}
\left(
\operatorname{Serialize}_{G_T}(m_G)
\right).
\end{equation}

The vendor-specific aggregate ciphertext is then

\begin{equation}
C_{v_1}
=
(c_1,c_2,c_3),
\end{equation}

where

\begin{equation}
c_1=tg_{v_1,1},
\qquad
c_2=t\left(PK_{v_1}+K_{S_{c,v_1},v_1}\right),
\end{equation}

and

\begin{equation}
c_3
=
m_G\cdot
\hat e
\left(
g_{v_1,4},
t g_{v_1,1}
\right),
\end{equation}

for fresh $t\in\mathbb{Z}_p^*$. Each authorized device uses its
individual decryption key together with the corresponding
board-specific aggregate term

\begin{equation}
b_{i,S_{c,v_1}}
=
\sum_{\substack{j\in S_{c,v_1}\\j\neq i}}
g_{v_1,n_{v_1}+1-j+i}
\end{equation}

to recover $m_G$ and subsequently derive $K_{\mathrm{AES}}$.

For this example, the required device-specific terms are shown in
Table~\ref{tab:example}.

\begin{table}[t]
	\centering
	\caption{Device-specific aggregate-decryption terms for
	$S_{c,v_1}=\{1,2,4\}$, $n_{v_1}=4$.}
	\label{tab:example}
	\begin{tabular}{@{}cccc@{}}
		\toprule
		\textbf{Device} &
		\textbf{Authorized} &
		\textbf{Decryption key} &
		\textbf{Aggregate term} \\
		\midrule
		$\mathrm{FPGA}_{v_1,1}$ & Yes &
		$d_{v_1,1}$ & $g_{v_1,4}+g_{v_1,2}$ \\
		$\mathrm{FPGA}_{v_1,2}$ & Yes &
		$d_{v_1,2}$ & $g_{v_1,6}+g_{v_1,3}$ \\
		$\mathrm{FPGA}_{v_1,3}$ & No &
		$d_{v_1,3}$ & --- \\
		$\mathrm{FPGA}_{v_1,4}$ & Yes &
		$d_{v_1,4}$ & $g_{v_1,8}+g_{v_1,7}$ \\
		\bottomrule
	\end{tabular}
\end{table}

The example also illustrates that aggregate ciphertext size does not
depend on the number of authorized devices. Although the auxiliary
decryption term differs for each device, the same vendor-specific
ciphertext is used by all authorized devices, while an unauthorized
device does not satisfy the corresponding aggregate-decryption relation.

\paragraph{Extension to a second vendor.}
The deployment is next extended to include
$\mathrm{FPGA}_{v_2,1}$ from a second vendor $v_2$. Vendor $v_2$ maintains
an independent aggregate-encryption domain with its own public parameters,
device keys, and master secret $\gamma_{v_2}$. The overall capability
cluster becomes

\[
S_c=S_{c,v_1}\cup S_{c,v_2},
\qquad
S_{c,v_2}=\{1\}.
\]

The authorization state is therefore represented by the set of
vendor-specific components

\begin{equation}
K_{S_c}
=
\left\{
K_{S_{c,v_1},v_1},
K_{S_{c,v_2},v_2}
\right\}.
\end{equation}

Each vendor domain is processed independently. No aggregate-encryption
secret or device decryption key is shared between $v_1$ and $v_2$, and
one aggregate-encryption operation is performed for each participating
vendor domain. Hence, adding a new vendor increases the number of
vendor-domain encryption operations without changing the ciphertext
structure within an individual domain.

\paragraph{Dynamic membership.}
Suppose $\mathrm{FPGA}_{v_1,3}$ is subsequently added to the deployment.
Only the authorization component associated with vendor $v_1$ is updated:

\begin{equation}
K_{S_{c,v_1}\cup\{3\},v_1}
=
K_{S_{c,v_1},v_1}
+
g_{v_1,2},
\end{equation}

where $g_{v_1,2}$ is the contribution of the newly added device for
$n_{v_1}=4$. The component associated with vendor $v_2$ remains unchanged.
The updated authorization state forms a new membership epoch and is used
for subsequently generated ciphertexts. Ciphertexts generated under the
previous epoch retain their original authorization state.

Consequently, a device revoked before a new epoch cannot recover
deployments generated under that new epoch, whereas a device enrolled in a
later epoch cannot recover ciphertexts generated under an earlier epoch.
Membership changes therefore affect authorization state without requiring
regeneration of the identity credentials of unaffected devices.

\paragraph{Authenticated and device-bound deployment.}
For a subsequent deployment, the selected target device first authenticates
to the registered IP provider using its vendor-scoped identity credential.
The authentication exchange establishes the shared session value
$K_{v_1,1,p}$ and the session key

\begin{equation}
K_{\mathrm{sess},v_1,1}
=
H_3
\left(
K_{v_1,1,p}
\parallel
ID_{v_1,1}
\parallel
ID_p^{(v_1)}
\parallel
N_{v_1,1}
\parallel
\tau_{v_1,1}
\right).
\end{equation}

The provider uses this session key to authenticate the
vendor-specific aggregate ciphertext and its deployment context through
the session-binding tag. The target device verifies the session-binding
tag before performing aggregate decryption.

After successful authorization, the device recovers
$\widehat{m}_G$ and derives

\begin{equation}
\widehat{K}_{\mathrm{AES}}
=
\mathrm{SHA256}
\left(
\operatorname{Serialize}_{G_T}(\widehat{m}_G)
\right).
\end{equation}

For an authorized device, $\widehat{m}_G=m_G$ and hence
$\widehat{K}_{\mathrm{AES}}=K_{\mathrm{AES}}$. The device then verifies
the provider's identity-bound signature for its assigned
architecture-specific bitstream variant and checks the corresponding
AES-256-GCM authentication tag. Only after all authentication,
authorization, signature, and integrity checks succeed is the bitstream
decrypted and accepted for FPGA configuration.

The same sequence is independently executed for devices belonging to
vendor $v_2$ under that vendor's aggregate-encryption and identity
domains. Thus, credentials and authorization material associated with one
vendor domain cannot be directly reused within another vendor domain.

\subsection{Capability-Aware Clustering and Heterogeneous Bitstream Deployment}
\label{sec:cap_heterogeneous}

SAEID separates deployment compatibility from vendor identity by organizing
eligible FPGA devices into capability-aware clusters. For a capability
requirement $c$, define

\begin{equation}
S_c =
\left\{
\mathrm{FPGA}_{v,i} :
\mathrm{Capability}(\mathrm{FPGA}_{v,i}) = c
\right\},
\label{eq:capability_cluster_def}
\end{equation}

where $c$ captures the hardware and deployment requirements of the target
accelerator, including FPGA-family compatibility, supported features,
available logic, memory, DSP resources, and other application-specific
constraints. Devices from different vendors may therefore belong to the
same capability cluster; capability determines eligibility, while vendor
identity determines the provisioning and cryptographic authorization
domain.

For each participating vendor $v$, let

\begin{equation}
S_{c,v}
=
\left\{
\mathrm{FPGA}_{v,i}\in S_c
\right\},
\qquad
I_{c,v}
=
\left\{
i:\mathrm{FPGA}_{v,i}\in S_{c,v}
\right\},
\end{equation}

where $I_{c,v}\subseteq\{1,\ldots,n_v\}$ is the corresponding device-index
set. The participating vendor domains are

\begin{equation}
\mathcal{V}_{S_c}
=
\left\{
v\in\mathcal{V}:S_{c,v}\neq\emptyset
\right\}.
\label{eq:vendor_subset}
\end{equation}

Each vendor $v$ establishes an independent aggregate-encryption domain
over a bilinear group $G$ of prime order $p$ with generator $g$ and
pairing $\hat{e}:G\times G\rightarrow G_T$. The vendor selects a setup
exponent $\alpha_v\in\mathbb{Z}_p^{*}$ and generates the public group
elements
\begin{equation}
g_{v,k}=g^{\alpha_v^k},
\qquad
k\in\{1,\ldots,2n_v\},
\quad k\neq n_v+1,
\label{eq:public_group_elements}
\end{equation}
where $n_v$ is the number of provisioned devices in vendor domain $v$.
The vendor independently selects the aggregate-encryption secret
$\gamma_v\in\mathbb{Z}_p^{*}$ and publishes the corresponding public
aggregate-encryption parameter
\begin{equation}
PK_v=\gamma_v g,
\label{eq:aggregate_public_parameter}
\end{equation}
where $g$ is the generator of $G$. The remaining public system
parameters, including the bilinear group description and the public
elements $\{g_{v,k}\}$, are denoted collectively by $\mathit{param}_v$.
The vendor retains $\gamma_v$ as the aggregate-encryption secret.
Each registered
device $\mathrm{FPGA}_{v,i}$ is provisioned with its individual decryption
key
\begin{equation}
d_{v,i}=\gamma_v g_{v,i}.
\label{eq:device_decryption_key}
\end{equation}

Each vendor independently constructs the aggregate-authorization component
for its selected devices following the pairing-based key-aggregate
construction used in AgEID~\cite{debnath2024ageid}:
\begin{equation}
K_{S_{c,v},v}
=
\sum_{j\in I_{c,v}}
g_{v,n_v+1-j}.
\label{eq:aggregate_authorization}
\end{equation}

The resulting multi-vendor authorization structure is

\begin{equation}
\mathcal{K}_{S_c}
=
\left\{
K_{S_{c,v},v}:v\in\mathcal{V}_{S_c}
\right\}.
\label{eq:multivendor_aggregate_key}
\end{equation}

Thus, a deployment may span multiple vendor domains without merging
their cryptographic domains. Capability membership is an eligibility
criterion, not a cryptographic security boundary.

Capability compatibility does not imply bitstream compatibility. Capability
compatibility determines whether a device is eligible for the deployment,
whereas architectural compatibility determines which compiled bitstream
variant can be executed by that device. Let

\begin{equation}
\mathcal{A}(S_c)
=
\left\{
\mathrm{Arch}(v,i):
\mathrm{FPGA}_{v,i}\in S_c
\right\},
\qquad
\mu=|\mathcal{A}(S_c)|.
\label{eq:architecture_classes}
\end{equation}

denote the distinct target architecture classes. The IP provider prepares
one compiled bitstream variant $\mathrm{Bit}^{(k)}$ for each
$k\in\{1,\ldots,\mu\}$, with devices in the same architecture class using
the corresponding variant. Hence, bitstream preparation scales with
$\mu$, rather than directly with the number of devices.

The IP provider protects all architecture-specific bitstream variants
using a common deployment AES key $m$. The key $m$ is then protected for
the selected devices through the multi-vendor authorization structure
$\mathcal{K}_{S_c}$, with one aggregate ciphertext generated for each
participating vendor domain. The corresponding encryption and individual
recovery procedures are given in Section~\ref{sec:aggregate_workflow}.

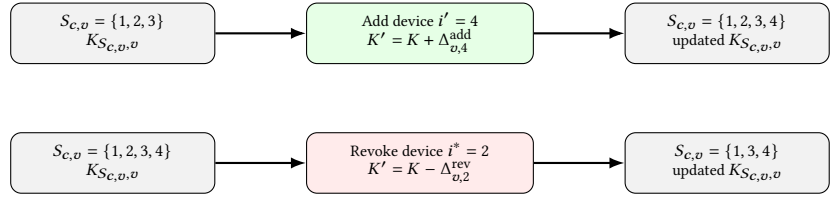
\begin{figure}[t]
\centering
\begin{tikzpicture}[
node distance=6mm and 9mm,
box/.style={
    draw,
    rounded corners,
    align=center,
    minimum height=0.8cm,
    font=\scriptsize
},
arr/.style={-{Latex[length=2mm]}, thick}
]

\node[box, fill=gray!10, minimum width=2.7cm] (before)
    {$S_{c,v}=\{1,2,3\}$\\$K_{S_{c,v},v}$};

\node[box, fill=green!10, right=1.2cm of before, minimum width=3.0cm]
    (add)
    {Add device $i'=4$\\
    $K' = K+\Delta^{\mathrm{add}}_{v,4}$};

\node[box, fill=gray!10, right=1.2cm of add, minimum width=2.7cm]
    (after1)
    {$S_{c,v}=\{1,2,3,4\}$\\
    updated $K_{S_{c,v},v}$};

\draw[arr] (before) -- (add);
\draw[arr] (add) -- (after1);

\node[box, fill=gray!10, below=9mm of before, minimum width=2.7cm]
    (before2)
    {$S_{c,v}=\{1,2,3,4\}$\\
    $K_{S_{c,v},v}$};

\node[box, fill=red!8, right=1.2cm of before2, minimum width=3.0cm]
    (revoke)
    {Revoke device $i^\ast=2$\\
    $K' = K-\Delta^{\mathrm{rev}}_{v,2}$};

\node[box, fill=gray!10, right=1.2cm of revoke, minimum width=2.7cm]
    (after2)
    {$S_{c,v}=\{1,3,4\}$\\
    updated $K_{S_{c,v},v}$};

\draw[arr] (before2) -- (revoke);
\draw[arr] (revoke) -- (after2);

\end{tikzpicture}
\caption{Dynamic membership update within a vendor-specific authorization
domain. Device addition or revocation modifies only the affected
aggregate-authorization component; identity credentials of unaffected
devices and authorization components of other vendors remain unchanged.}
\label{fig:membership}
\end{figure}

\subsection{Multi-Vendor, Multi-Provider Deployment and Dynamic Membership}
\label{sec:multi_vendor_dynamic}

The vendor-specific authorization components defined for a
capability-aware cluster form the authorization structure for a
multi-vendor deployment:
\begin{equation}
\mathcal{K}_{S_c}
=
\left\{
K_{S_{c,v},v} :
v\in\mathcal{V}_{S_c}
\right\}.
\label{eq:deployment_auth_structure}
\end{equation}

Each component is constructed independently within its vendor domain.
Thus, devices from different vendors may participate in the same
deployment without sharing vendor master secrets, device private keys, or
identity-authority secrets.

SAEID also supports deployments from multiple IP providers. For a provider
$p$, the selected devices and corresponding vendor domains define the
authorization structure for that deployment. The provider independently
establishes the deployment session and generates a fresh AES key $m$.
The resulting deployment is therefore associated with its own
authorization state, session context, and deployment key, without
requiring changes to the vendor-issued device credentials.

\paragraph{Dynamic Membership.}

As illustrated in Fig.~\ref{fig:membership}, SAEID supports device addition
and revocation without reconstructing the complete multi-vendor
authorization structure.  
For an already provisioned
device $i'$ added to $S_{c,v}$,

\begin{equation}
K_{S_{c,v}\cup\{i'\},v}
=
K_{S_{c,v},v}
+
\Delta^{\mathrm{add}}_{v,i'},
\label{eq:add_device}
\end{equation}

where

\begin{equation}
\Delta^{\mathrm{add}}_{v,i'}
=
g_{v,n_v+1-i'}.
\end{equation}

Similarly, revocation of device $i^\ast$ is performed as

\begin{equation}
K_{S_{c,v}\setminus\{i^\ast\},v}
=
K_{S_{c,v},v}
-
\Delta^{\mathrm{rev}}_{v,i^\ast},
\label{eq:revoke_device}
\end{equation}

with

\begin{equation}
\Delta^{\mathrm{rev}}_{v,i^\ast}
=
g_{v,n_v+1-i^\ast}.
\end{equation}

Hence, membership changes update only the affected vendor-specific
aggregate component and have $O(1)$ aggregate-update cost. This cost
assumes that an added device has already completed vendor provisioning and
received its decryption and identity keys; provisioning and
identity-key generation are outside the membership-update measurement.
Revocation likewise does not require re-provisioning unaffected devices.

Each aggregate authorization state is associated with a membership epoch
$e$, denoted by $K_{S_{c,v}}^{(e)}$. A membership change creates a new
epoch, and subsequently generated ciphertexts use the updated
authorization state. Thus, devices added in epoch $e+1$ cannot decrypt
ciphertexts from epoch $e$, while devices revoked before epoch $e+1$
cannot decrypt newly generated ciphertexts under epoch $e+1$.
Previously issued ciphertexts are not retroactively re-encrypted.

Identity keys authenticate devices, whereas the aggregate authorization
state determines deployment access. Hence, membership changes do not
require modification of the identity keys of unaffected devices.

\begin{figure}[t]
\centering
\begin{tikzpicture}[
  node distance=5mm and 7mm,
  box/.style={
    draw,
    rounded corners,
    align=center,
    minimum height=0.8cm,
    font=\scriptsize
  },
  arr/.style={-{Latex[length=2mm]}, thick}
]

\node[box, fill=orange!12] (setup)
    {Vendor identity setup\\
     $\mathrm{ExtractIdentityKey}$};

\node[box, fill=purple!10, right=of setup] (chal)
    {AuthenticateDevice\\
     Challenge $(N,\tau)$};

\node[box, fill=purple!10, right=of chal] (resp)
    {Device response\\
     $(U,V)$};

\node[box, fill=purple!10, below=of resp] (verify)
    {VerifyAuthentication\\
     $\Rightarrow K_{\mathrm{sess}}$};

\node[box, fill=blue!8, left=of verify] (extract)
    {Extract($S_{c,v}$)\\
     $\to K_{S_{c,v},v}$};

\node[box, fill=blue!8, left=of extract] (encrypt)
    {SAEID.Encrypt\\
     $\to C_v=(c_1,c_2,c_3)$};

\node[box, fill=gray!12, below=of encrypt] (bind)
    {Session binding\\
     $T_v=\mathrm{MAC}_{K_{\mathrm{sess}}}(D_v)$};

\node[box, fill=gray!12, below=of bind] (upload)
    {CSP relay/storage\\
     $\{C_v,T_v,\mathrm{Ct}^{(k)},\mathrm{Tag}_{\mathrm{GCM}}^{(k)},
     \sigma_{p,v,i}^{(k)}\}$};

\node[box, fill=green!10, right=of upload] (verifytag)
    {Verify $T_v$\\
     with $K_{\mathrm{sess}}$};

\node[box, fill=green!10, right=of verifytag] (decrypt)
    {SAEID.Decrypt\\
     $\to \hat m$};

\node[box, fill=green!10, right=of decrypt] (verifybs)
    {Verify signature\\
     \& GCM tag};

\node[box, fill=green!10, right=of verifybs] (bitdec)
    {Bitstream\\
     decrypt \& execute};

\draw[arr] (setup) -- (chal);
\draw[arr] (chal) -- (resp);
\draw[arr] (resp) -- (verify)
    node[midway,right,font=\tiny]{$(U,V)$};

\draw[arr] (verify) -- (extract)
    node[midway,above,font=\tiny]{on accept};

\draw[arr] (extract) -- (encrypt)
    node[midway,above,font=\tiny]{$K_{S_{c,v},v}$};

\draw[arr] (encrypt) -- (bind)
    node[midway,right,font=\tiny]{$C_v$};

\draw[arr] (bind) -- (upload)
    node[midway,right,font=\tiny]{$T_v$};

\draw[arr] (upload) -- (verifytag)
    node[midway,above,font=\tiny]{encrypted artifacts};

\draw[arr] (verifytag) -- (decrypt)
    node[midway,above,font=\tiny]{on accept};

\draw[arr] (decrypt) -- (verifybs)
    node[midway,above,font=\tiny]{$\hat m$};

\draw[arr] (verifybs) -- (bitdec)
    node[midway,above,font=\tiny]{on accept};

\end{tikzpicture}%
\caption{End-to-end authentication, aggregate-authorization, and secure
bitstream deployment workflow. Device authentication establishes the
session key $K_{\mathrm{sess}}$, which is used to authenticate the
vendor-specific aggregate ciphertext through the session-binding tag
$T_v$. The CSP relays only encrypted deployment artifacts and
protocol metadata. At the target FPGA, the session-binding tag is verified
before aggregate content-key recovery, followed by identity-bound
signature verification and AES-GCM tag verification before bitstream
decryption and configuration.}
\label{fig:enc-workflow}
\end{figure}
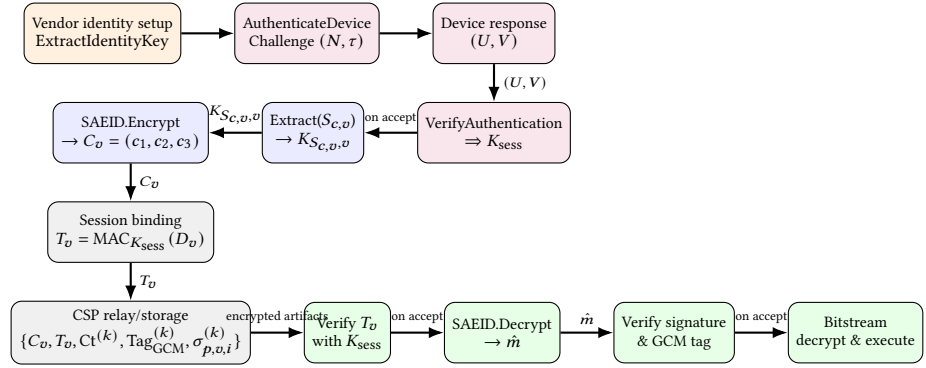

\subsection{Aggregate Encryption and Secure Bitstream Execution}
\label{sec:aggregate_workflow}

SAEID retains the central deployment principle of AgEID: aggregate
encryption protects a group-valued deployment secret from which the
symmetric key used for bitstream protection is derived, rather than
protecting the bitstream payload itself. For a capability-aware deployment
cluster $S_c$ with $\mu$ architecture-specific variants
$\{\mathrm{Bit}^{(k)}\}_{k=1}^{\mu}$, the IP provider samples a fresh
group-valued deployment secret
$m_G\in G_T$. The corresponding AES-256 deployment key is deterministically
derived from the canonical byte serialization of $m_G$ as
\begin{equation}
K_{\mathrm{AES}}
=
\mathrm{SHA256}\!\left(
\operatorname{Serialize}_{G_T}(m_G)
\right),
\label{eq:aes-key-derivation}
\end{equation}
where $K_{\mathrm{AES}}\in\{0,1\}^{256}$ is used as the key for
AES-256-GCM. The session key $K_{\mathrm{sess},v,i}$ is separate from
$K_{\mathrm{AES}}$ and is used only to bind the authenticated
provider--device exchange. \autoref{fig:enc-workflow} shows the end-to-end
authentication, aggregate-authorization, and secure bitstream deployment
workflow.

For each participating vendor $v\in\mathcal{V}_{S_c}$, the
vendor-specific aggregate-authorization component $K_{S_{c,v},v}$
protects the group-valued deployment secret $m_G$. Let
$t\xleftarrow{R}\mathbb{Z}_p$ be a fresh encryption scalar. The aggregate
ciphertext is constructed using the bilinear pairing
$\hat{e}:G\times G\rightarrow G_T$ inherited from the AgEID
construction~\cite{debnath2024ageid}. The ciphertext is
\begin{equation}
C_v=(c_1,c_2,c_3),
\end{equation}
where
\begin{align}
c_1 &= tg,\\
c_2 &= t(PK_v+K_{S_{c,v},v}),\\
c_3 &= m_G\cdot
\hat e(g_{v,n_v},tg_{v,1}).
\end{align}
Here, $PK_v$ denotes the public parameter of the vendor's
aggregate-encryption domain. Hence, one aggregate-encryption operation is
performed per participating vendor domain rather than per authorized FPGA
device.

For each architecture-specific variant, the provider encrypts and
authenticates the bitstream using the common AES deployment key
$K_{\mathrm{AES}}$:
\begin{equation}
\left(
\mathrm{Ct}^{(k)},
\mathrm{Tag}_{\mathrm{GCM}}^{(k)}
\right)
\leftarrow
\mathrm{AES\text{-}256\text{-}GCM.Encrypt}
\left(
K_{\mathrm{AES}},IV^{(k)},\mathrm{Bit}^{(k)},\mathrm{Hdr}^{(k)}
\right).
\label{eq:bitstream_aead}
\end{equation}
Here, $IV^{(k)}$ is the GCM initialization vector (nonce), which is
unique for each encryption performed under the same key
$K_{\mathrm{AES}}$. The value $\mathrm{Hdr}^{(k)}$ denotes the associated
data authenticated by AES-GCM. The provider additionally generates an
identity-bound signature that binds the deployment to the authenticated
provider, target device, and architecture-specific variant.

Following successful device authentication, the provider and device
establish the session key $K_{\mathrm{sess},v,i}$ for the authenticated
deployment exchange. Its establishment is described in
Section~\ref{sec:auth_bitstream}. The resulting session key is used to
authenticate the deployment metadata and bind the aggregate ciphertext
to the authenticated provider--device session.

\paragraph{Device-Side Key Recovery.}
For a device $\mathrm{FPGA}_{v,i}\in S_{c,v}$, define
\begin{equation}
b_{i,S_{c,v}}=
\sum_{\substack{j\in I_{c,v}\\j\neq i}}
g_{v,n_v+1-j+i},
\label{eq:device_recovery_term}
\end{equation}
where $I_{c,v}$ is the corresponding set of device indices. Using its
individual decryption key $d_{v,i}$, the device recovers the
group-valued deployment secret as
\begin{equation}
\widehat{m}_G
=
c_3\cdot
\frac{
\hat{e}\left(d_{v,i}+b_{i,S_{c,v}},c_1\right)
}{
\hat{e}\left(g_{v,i},c_2\right)
}.
\label{eq:saeid_decryption}
\end{equation}
For an authorized device $i\in I_{c,v}$, the bilinear cancellation
property gives
\begin{equation}
\widehat{m}_G=m_G.
\end{equation}
The device then derives the AES deployment key from the recovered
group element:
\begin{equation}
\widehat{K}_{\mathrm{AES}}
=
\mathrm{SHA256}\!\left(
\operatorname{Serialize}_{G_T}(\widehat{m}_G)
\right).
\label{eq:recovered-aes-key}
\end{equation}
Consequently, for an authorized device,
$\widehat{K}_{\mathrm{AES}}=K_{\mathrm{AES}}$, whereas devices outside
the authorized subset do not satisfy the corresponding aggregate
decryption condition. The algebraic correctness of
Eq.~\eqref{eq:saeid_decryption} is provided in
Appendix~\ref{app:proofs}.

The complete deployment procedure is summarized in
Algorithm~\ref{alg:aggregate_execution} and illustrated in
\autoref{fig:enc-workflow}. The CSP relays and stores only encrypted
artifacts and protocol metadata. At the target FPGA, the device first
verifies the session-binding tag, recovers the group-valued deployment
secret, derives the AES deployment key, verifies the identity-bound
signature and AES-GCM authentication tag, and only then decrypts and
configures the assigned bitstream variant.

\begin{algorithm}[t]
\caption{Aggregate Encryption and Secure Bitstream Execution}
\label{alg:aggregate_execution}
\begin{algorithmic}[1]
\Require Capability cluster $S_c$, architecture-specific variants
$\{\mathrm{Bit}^{(k)}\}_{k=1}^{\mu}$

\Ensure Secure deployment of the assigned bitstream variant

\State Generate a fresh group-valued deployment secret
$m_G\gets_R G_T$
\Comment{Deployment-secret generation}

\State Derive
$K_{\mathrm{AES}}\gets
\mathrm{SHA256}(\operatorname{Serialize}_{G_T}(m_G))$
\Comment{AES-256 deployment-key derivation}

\ForAll{$k=1,\ldots,\mu$}
    \State Compute
    $\left(\mathrm{Ct}^{(k)},\mathrm{Tag}_{\mathrm{GCM}}^{(k)}\right)
    \gets
    \mathrm{AES\text{-}256\text{-}GCM.Encrypt}
    \left(K_{\mathrm{AES}},IV^{(k)},\mathrm{Bit}^{(k)},\mathrm{Hdr}^{(k)}\right)$
    \Comment{Bitstream encryption and authentication}
\EndFor

\ForAll{$v\in\mathcal{V}_{S_c}$}

    \State Select fresh $t\xleftarrow{R}\mathbb{Z}_p$

    \State Compute
    $C_v=(c_1,c_2,c_3)$ using $K_{S_{c,v},v}$ and $m_G$
    \Comment{Aggregate encryption of $m_G$}

    \ForAll{$i\in I_{c,v}$}
        \State Establish the authenticated session key
        $K_{\mathrm{sess},v,i}$
        \Comment{Provider--device authentication}

        \State Form
        $D_{v,i}\gets
        \mathrm{ID}_{v,i}\parallel
        \mathrm{ID}_{p}^{(v)}\parallel
        N_{v,i}\parallel
        \tau_{v,i}\parallel
        v\parallel c\parallel C_v$

        \State Compute
        $T_{v,i}\gets
        \mathrm{MAC}_{K_{\mathrm{sess},v,i}}(D_{v,i})$
        \Comment{Session binding}
    \EndFor

    \State Relay
    $\{C_v,T_{v,i}\}_{i\in I_{c,v}}$ and the assigned protected
    variants through the CSP
    \Comment{CSP performs relay/storage only}

    \ForAll{$i\in I_{c,v}$}
        \State Verify $T_{v,i}$ using $K_{\mathrm{sess},v,i}$
        \If{verification fails}
            \State Reject deployment
        \EndIf

        \State Recover $\widehat{m}_G$ from $C_v$ using individual decryption
        \Comment{Aggregate decryption}

        \State Derive
        $\widehat{K}_{\mathrm{AES}}\gets
        \mathrm{SHA256}(\operatorname{Serialize}_{G_T}(\widehat{m}_G))$
        \Comment{AES-256 key recovery}

        \State Verify the identity-bound provider signature
        \Comment{Provider/deployment authentication}

        \State Verify $\mathrm{Tag}_{\mathrm{GCM}}^{(k)}$
        \Comment{Bitstream integrity verification}

        \If{any verification fails}
            \State Reject deployment
        \Else
            \State Decrypt and configure the assigned $\mathrm{Ct}^{(k)}$
            using $\widehat{K}_{\mathrm{AES}}$
        \EndIf
    \EndFor

\EndFor

\end{algorithmic}
\end{algorithm}

\begin{figure}[t]
\centering
\begin{tikzpicture}[
  node distance=8mm and 14mm,
  actor/.style={font=\scriptsize\bfseries},
  msg/.style={-{Latex[length=2mm]}, thick},
  lifeline/.style={dashed, gray}
]
\node[actor] (dev) at (0,0) {FPGA Device $v,i$};
\node[actor] (ipp) at (6,0) {IP Provider};
\draw[lifeline] (dev) -- ++(0,-6.2);
\draw[lifeline] (ipp) -- ++(0,-6.2);

\draw[msg] (ipp.south) ++(0,-0.6) -- node[above,font=\tiny]{1. Challenge $(N,\tau)$} ++(-6,0) node[left]{};
\node[font=\tiny] at (3,-1.0) {$N \in_R \{0,1\}^k,\ \tau = \mathrm{timestamp}$};

\node[font=\tiny] at (6,-4.0) {IP Provider verifies the identity-based proof};
\node[font=\tiny] at (0,-2.8) {$V=(r+h)\mathrm{SK}_{\mathrm{ID}_{v,i}}$};

\draw[msg] (0,-3.3) -- node[above,font=\tiny]{2. Response $(U,V)$} (6,-3.3);

\draw[msg] (6,-5.0) -- node[above,font=\tiny]{3. Verification result} (0,-5.0);
\node[font=\tiny] at (3,-5.6)
{Accept iff proof verification succeeds and $\tau$ is fresh};

\draw[msg] (6,-6.0) -- node[above,font=\tiny]{4. Session-bound authorization $(C_v,T_v)$} (0,-6.0);
\end{tikzpicture}%
\caption{Authentication protocol message flow: challenge generation,
identity-based device response, verification, and session-bound vendor
authorization. Successful authentication establishes the session key
$K_{\mathrm{sess}}$, which is subsequently used to authenticate the
vendor-specific aggregate ciphertext $C_v$ through the session-binding tag
$T_v$. Detailed authentication equations are given in
Section~\ref{sec:auth_bitstream}.}
\label{fig:auth-protocol}
\end{figure}
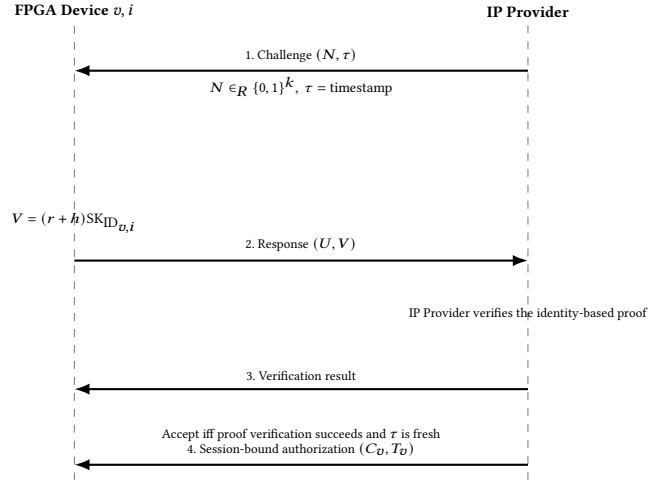

\subsection{Integrated Identity-Based Authentication and Bitstream Integrity}
\label{sec:auth_bitstream}

SAEID addresses two distinct security requirements in FPGA-cloud
deployment. Aggregate encryption provides scalable authorization of the
AES deployment key $m$, whereas identity-based cryptography provides
certificate-free authentication of FPGA devices and registered IP
providers. SAEID integrates these two mechanisms within separate
vendor-specific cryptographic domains: the aggregate-encryption secret
$\gamma_v$ protects deployment-key authorization, while the identity
secret $s_v$ supports device authentication and provider
authentication/signatures. Thus, identity-based authentication complements
rather than replaces the aggregate-encryption mechanism inherited from
AgEID.

For each vendor $v$, the identity framework operates over a cyclic group
$\mathbb{G}$ of prime order $p$ with generator $g$ and bilinear pairing
\[
\hat{e}:\mathbb{G}\times\mathbb{G}\rightarrow\mathbb{G}_T.
\]
The vendor operates the identity authority $\mathcal{IA}_v$, which
maintains the master identity secret
$s_v\in\mathbb{Z}_p^{*}$ and publishes the corresponding public identity
parameter
\begin{equation}
P_{\mathrm{pub},v}=s_vg.
\label{eq:identity_public}
\end{equation}
The identity domain is independent of the aggregate-encryption domain,
whose vendor-specific secret is $\gamma_v$.

During vendor provisioning, each FPGA device $\mathrm{FPGA}_{v,i}$ is
assigned a unique identity
\begin{equation}
\mathrm{ID}_{v,i}
=
\texttt{DEV:}\parallel
\mathrm{VendorID}_{v}\parallel
\mathrm{DeviceIndex}_{i}.
\label{eq:device_identity}
\end{equation}
Its identity point and corresponding private key are
\begin{equation}
Q_{v,i}=H_1(\mathrm{ID}_{v,i}),
\qquad
\mathrm{SK}_{\mathrm{ID}_{v,i}}
=
s_vQ_{v,i},
\label{eq:device_identity_key}
\end{equation}
where
\[
H_1:\{0,1\}^{*}\rightarrow\mathbb{G}
\]
maps an identity string to a point in $\mathbb{G}$. The resulting
identity private key is provisioned to the corresponding FPGA and is
used for identity-based authentication.

Similarly, each IP provider registered under vendor $v$ is assigned a
vendor-scoped identity from a separate namespace:
\begin{equation}
\mathrm{ID}_{p}^{(v)}
=
\texttt{PROV:}\parallel
\mathrm{VendorID}_{v}\parallel
\mathrm{ProviderID}_{p}.
\label{eq:provider_identity}
\end{equation}
The corresponding identity point and private key are
\begin{equation}
Q_{p}^{(v)}
=
H_1(\mathrm{ID}_{p}^{(v)}),
\qquad
\mathrm{SK}_{\mathrm{ID}_{p}}^{(v)}
=
s_vQ_{p}^{(v)}.
\label{eq:provider_identity_key}
\end{equation}

Thus, device and provider identities are derived within the same
vendor-specific identity domain but use distinct namespaces. This allows
SAEID to authenticate both the target FPGA device and the registered IP
provider without introducing a separate certificate-management
infrastructure, while preserving the separation of identity-based
authentication from aggregate authorization.

\paragraph{Device Authentication.}

The provider initiates authentication with a fresh nonce and timestamp,
\begin{equation}
N_{v,i}\xleftarrow{R}\{0,1\}^{k},
\qquad
\tau_{v,i}=\mathrm{Timestamp}.
\label{eq:auth_challenge}
\end{equation}

The device samples a random scalar
$r\xleftarrow{R}\mathbb{Z}_p$ and computes the randomized identity
point $U$, the challenge scalar $h$, and the response $V$ as
\begin{align}
U &= rQ_{v,i}, \label{eq:auth_U}\\
h &= H_2(
\mathrm{ID}_{v,i}\parallel N_{v,i}\parallel\tau_{v,i}\parallel U),
\label{eq:auth_h}\\
V &= (r+h)\mathrm{SK}_{\mathrm{ID}_{v,i}}.
\label{eq:auth_V}
\end{align}
Here, $U$ is the randomized proof component and $V$ is the corresponding
identity-based response, while $h$ is the verifier challenge derived from
the device identity, freshness values, and $U$.

The provider accepts the proof only if $\tau_{v,i}$ is fresh and
\begin{equation}
\hat e(g,V)
\stackrel{?}{=}
\hat e
\left(
P_{\mathrm{pub},v},
U+hQ_{v,i}
\right).
\label{eq:auth_verify}
\end{equation}

For an honestly generated proof, correctness follows from bilinearity and
$\mathrm{SK}_{\mathrm{ID}_{v,i}}=s_vQ_{v,i}$. Following successful
authentication, the device and provider derive the shared pairing value
using their respective vendor-scoped identity credentials:
\begin{equation}
K_{v,i,p}
=
\hat e\!\left(
\mathrm{SK}_{\mathrm{ID}_{v,i}},Q^{(v)}_{p}
\right)
=
\hat e\!\left(
Q_{v,i},\mathrm{SK}^{(v)}_{\mathrm{ID}_{p}}
\right).
\label{eq:shared_pairing}
\end{equation}
The session key is then derived as
\begin{equation}
K_{\mathrm{sess},v,i}
=
H_3\!\left(
K_{v,i,p}
\parallel ID_{v,i}
\parallel ID^{(v)}_{p}
\parallel N_{v,i}
\parallel \tau_{v,i}
\right).
\label{eq:session_key}
\end{equation}

The session key is used for deployment-session binding and remains
distinct from the AES key $m$ protected by aggregate encryption. The
authentication exchange is illustrated in~\autoref{fig:auth-protocol}. 

\paragraph{Identity-Bound Bitstream Integrity.}

For each architecture-specific encrypted variant $\mathrm{Ct}^{(k)}$, the
registered provider generates an identity-bound signature using its
vendor-scoped identity private key
$\mathrm{SK}_{\mathrm{ID}_{p}}^{(v)}$. Following the same
pairing-based Fiat--Shamir construction used for device authentication,
the signature is computed over the provider identity, target device
identity, encrypted bitstream variant, and associated release metadata.
Thus, the signature binds the protected variant to its registered provider
and intended FPGA device.

At the target FPGA, the signature is verified using the corresponding
vendor public identity parameter $P_{\mathrm{pub},v}$ and the registered
provider identity. The device then verifies the AES-GCM authentication tag
using the recovered AES key $\hat m$. Deployment is accepted only when
both provider authentication and encrypted-payload integrity are
successfully verified.

The detailed signature construction and verification equations, together
with the corresponding correctness and unforgeability arguments, are
provided in Appendix~\ref{app:proofs}.

\begin{algorithm}[t]
\caption{Identity Authentication and Bitstream Integrity Verification}
\label{alg:identity_integrity}
\begin{algorithmic}[1]
\Require Device $\mathrm{FPGA}_{v,i}$, provider $p$, encrypted variant
$\mathrm{Ct}^{(k)}$, recovered AES key $\hat m$
\Ensure Authenticated deployment or rejection

\State Provider selects fresh $(N_{v,i},\tau_{v,i})$
\State Provider sends $(N_{v,i},\tau_{v,i})$ to the device
\State Device samples $r\xleftarrow{R}\mathbb{Z}_p$ and computes
$U\gets rQ_{v,i}$
\State $h\gets
H_2(\mathrm{ID}_{v,i}\parallel N_{v,i}\parallel\tau_{v,i}\parallel U)$
\State $V\gets(r+h)\mathrm{SK}_{\mathrm{ID}_{v,i}}$
\State Device sends $(U,V)$ to the provider
\State Provider verifies freshness and checks
\[
\hat e(g,V)
\stackrel{?}{=}
\hat e(P_{\mathrm{pub},v},U+hQ_{v,i})
\]

\If{verification fails}
    \State Reject authentication
\EndIf

\State Provider generates the identity-bound signature
$\sigma_{p,v,i}^{(k)}$
\State Provider sends $\sigma_{p,v,i}^{(k)}$ with the encrypted variant
\State Device verifies $\sigma_{p,v,i}^{(k)}$

\If{signature verification fails}
    \State Reject bitstream
\EndIf

\State Device verifies the AES-GCM authentication tag using $\hat m$

\If{tag verification fails}
    \State Reject bitstream
\Else
    \State Accept and decrypt the assigned variant
\EndIf

\end{algorithmic}
\end{algorithm}

\section{Security Analysis}
\label{sec:security}

SAEID considers a probabilistic polynomial-time (PPT) adversary
$\mathcal{A}$ that may control the cloud service provider (CSP), observe,
replay, modify, reorder, or inject deployment messages, and compromise a
subset of FPGA devices. A compromised device exposes its locally stored decryption and identity keys. The adversary may therefore
attempt unauthorized recovery of the AES key $m$, device impersonation,
bitstream substitution, cross-vendor or cross-provider privilege
escalation, and exploitation of membership changes. Vendor master
secrets, identity-authority secrets, and protected secrets stored in
uncompromised devices are outside the adversary's capabilities, as
defined in Section~\ref{sec:threat}.

SAEID separates the aggregate-encryption domain from the identity domain.
The former is protected by the vendor-specific secret $\gamma_v$, while
device authentication and identity-bound provider signatures use the
identity secret $s_v$. The AES-GCM key $m$ is freshly generated for each
deployment, and the session key $K_{\mathrm{sess},v,i}$ is derived from
the authenticated device--provider exchange. The security proof sketches and the corresponding advantage bound are provided in Appendix~\ref{app:proofs}.

\paragraph{AES-Key Confidentiality and Collusion Resistance.}

SAEID applies aggregate encryption to the group-valued deployment secret
$m_G\in G_T$, from which the AES-256 deployment key
$K_{\mathrm{AES}}$ is derived. For a vendor-specific deployment set
$S_{c,v}$, only devices satisfying the corresponding aggregate-
authorization relation can recover $m_G$ and consequently derive
$K_{\mathrm{AES}}$. An unauthorized device, the CSP, or a coalition of
compromised devices that does not satisfy this relation cannot recover
$m_G$ from the aggregate ciphertext except with negligible probability
under the BDHE assumption underlying the aggregate-encryption
construction.

The confidentiality of the protected bitstream variants then follows
from the confidentiality of $K_{\mathrm{AES}}$ and the security of
AES-256-GCM. Importantly, possession of an identity credential does not
provide aggregate-decryption capability because the identity secret
$s_v$ and the aggregate-encryption secret $\gamma_v$ are distinct.

SAEID retains the aggregate-encryption algebra of AgEID; its extensions
change device selection, vendor-domain organization, authentication, and
deployment authorization without modifying the underlying
confidentiality mechanism. The complete reduction is provided in
Appendix~\ref{app:proofs}.

\paragraph{Dynamic Device Membership.}

SAEID supports device addition and revocation by updating only the
affected vendor-specific aggregate-authorization component according to
the addition and revocation procedures defined earlier in the construction.
The $O(1)$ complexity applies specifically to this aggregate-key update;
provisioning and installation of long-term credentials for a newly added
device are separate operations.

Each membership state is associated with an epoch
$E_0,E_1,\ldots$. A membership change creates a new epoch, and devices
added or revoked by that change are respectively included in or excluded
from the updated authorization state for subsequently generated
ciphertexts. Previously generated ciphertexts remain governed by their
original authorization state.

Hence, a revoked device cannot recover newly generated ciphertexts using
its previous authorization, and a newly enrolled device cannot recover
earlier ciphertexts solely because of its subsequent enrollment. These
forward- and backward-security properties apply to newly generated
authorization material and do not revoke plaintext keys or bitstreams
already legitimately obtained by a device.

Identity credentials of unaffected devices remain unchanged. Provider
credential revocation is outside the present dynamic-membership model,
which covers FPGA-device membership only. The membership epoch versions
the authorization state across membership changes and is distinct from
the nonce- and timestamp-based freshness mechanism used for authentication
replay resistance. The formal dynamic-membership argument is provided in
Appendix~\ref{app:proofs}.

\paragraph{Device Authentication and Replay Resistance.}

SAEID authenticates each target FPGA using a randomized identity-based
Fiat--Shamir proof. The provider supplies a fresh nonce $N_{v,i}$ and
timestamp $\tau_{v,i}$, and these values are incorporated into the
authentication challenge. Following successful verification, both
participants derive the session key as defined in
Eq.~\eqref{eq:session_key}, which binds the vendor-specific aggregate
ciphertext to the authenticated deployment context.

An adversary without the target device's identity private key cannot
produce an accepted proof for a fresh challenge except with negligible
probability under the stated identity-based Fiat--Shamir assumptions.
Previously recorded authentication transcripts are prevented from being
reused through fresh nonces and timestamps, together with rejection of
previously accepted freshness values.

The per-device freshness state is assumed to be maintained in trusted
persistent storage and preserved across ordinary power cycles and FPGA
reconfiguration. Loss or rollback of this trusted state is outside the
present threat model. Thus, the replay-resistance guarantee is explicitly
conditional on preservation of the trusted freshness state.

\paragraph{Bitstream Integrity and Device Binding.}

SAEID combines an identity-bound provider signature with AES-GCM
authenticated encryption. The signature challenge incorporates the
provider identity, target-device identity, encrypted payload, and release
metadata. Consequently, the signature is bound to a particular provider,
encrypted payload, and target FPGA rather than to a generic bitstream
release.

The target device accepts a protected variant only when both the
identity-bound signature and the AES-GCM authentication tag verify.
Modification of the encrypted payload invalidates the authenticated
encryption check, while redirection of a valid signature to another
device changes the identity-bound challenge and invalidates the signature.
Therefore, under identity-based signature unforgeability and AES-GCM
authenticity, modified, substituted, forged, or incorrectly targeted
bitstreams are rejected except with negligible probability.

\paragraph{Multi-Vendor and Multi-Provider Isolation.}

Each vendor maintains an independent aggregate-encryption domain and
vendor-scoped identity domain. A device credential or authorization
component from vendor $v$ therefore does not directly satisfy the
authentication or authorization conditions of another vendor $v'$. For
each provider deployment, the aggregate-authorization components and
session-binding keys are maintained within the corresponding vendor and
deployment contexts. Consequently, compromise of one provider deployment
does not by itself provide the cryptographic material required to access
another independent provider deployment or vendor domain, assuming the
corresponding vendor and provider credentials remain protected.

\paragraph{Protection Against a Partially Trusted CSP}

The CSP performs resource management, deployment storage, and message
relay, but does not receive vendor master secrets, identity-authority
secrets, device private keys, the AES key $m$, or
$K_{\mathrm{sess},v,i}$. Consequently, cryptographic confidentiality and
deployment authenticity do not depend on the CSP possessing these secret
values. The aggregate-encryption mechanism prevents recovery of $m$ by
the CSP, while provider authentication and AES-GCM verification allow the
target device to detect unauthorized or modified deployment artifacts.
Under the stated cryptographic and freshness assumptions, the CSP cannot
recover $m$ or cause a modified deployment artifact to pass the device-side
verification checks except with negligible probability.

\paragraph{Secure Composition.}

The SAEID security properties compose through distinct cryptographic
roles: $\gamma_v$ protects aggregate authorization of $m$; $s_v$ supports
device authentication and provider signatures; $K_{\mathrm{sess},v,i}$
provides session binding; and $m$ protects the bitstream through
AES-GCM. The identity-based authentication and provider-signature
mechanisms share the vendor identity domain $s_v$, but use separate
protocol roles and identity namespaces. Under the stated key-separation,
domain-separation, cryptographic, and freshness assumptions, an
end-to-end violation must correspond to a failure of at least one of the
underlying confidentiality, authentication, signature, authenticated-
encryption, or freshness mechanisms. The detailed game-based reductions
and advantage bounds are given in Appendix~\ref{app:proofs}.

\section{Scalability Analysis}
\label{sec:scalability}

Let $|S_c|$ denote the number of FPGA devices in a deployment cluster,
$|S_{c,v}|$ the number of selected devices belonging to vendor $v$,
$V$ the number of participating FPGA vendors, $P$ the number of active
IP providers, and $\mu$ the number of architecture-specific bitstream
variants required by a deployment. SAEID separates these factors across
different stages of deployment. Aggregate-authorization preparation
depends on the number of selected devices within each vendor domain,
aggregate content-key encryption depends on the number of participating
vendor domains, and heterogeneous bitstream preparation depends on the
number of distinct target architectures.

\subsection{Complexity of Core Operations}
\label{sec:complexity}

Vendor provisioning and long-term device-key generation are enrollment
operations and require $O(n_v)$ operations for $n_v$ devices in vendor
domain $v$. For a selected deployment subset $S_{c,v}$, aggregate-
authorization preparation requires $O(|S_{c,v}|)$ group operations,
after which aggregate encryption of the AES key $m$ requires $O(1)$
cryptographic operations per vendor domain and $O(V)$ for a deployment
spanning $V$ vendors. Device-side recovery is $O(1)$ when the
authorization-dependent decryption term is precomputed; computing this
term from the deployment set at runtime requires
$O(|S_{c,v}|)$ group additions. Device authentication and session binding
use a fixed number of cryptographic operations per device session and are
therefore $O(1)$ with respect to the deployment population.

Dynamic membership updates only the affected vendor-specific
aggregate-authorization component, following
Eqs.~\eqref{eq:add_device} and~\eqref{eq:revoke_device}.
Hence, device addition and revocation require only a local update to the
corresponding authorization state.

Thus, device addition and revocation require $O(1)$ aggregate-key update
work with respect to the deployment population. This bound excludes
device provisioning and long-term credential installation; unaffected
devices require no re-provisioning and previously generated ciphertexts
need not be re-encrypted.

Heterogeneous deployment introduces a separate factor. If $\mu$ distinct
target architectures occur in the selected cluster, the provider prepares
and authenticated-encrypts $\mu$ architecture-specific bitstream variants,
giving $O(\mu)$ variant-processing instances. The cost of compiling each
variant remains architecture- and application-dependent and is separate
from aggregate authorization.~\autoref{tab:end_to_end_complexity} shows the complexity of SAEID's deployment modules. 

\begin{table}[t]
\centering
\caption{End-to-End Asymptotic Complexity of SAEID Deployment}
\label{tab:end_to_end_complexity}
\footnotesize
\begin{tabular}{p{0.58\columnwidth}c}
\toprule
\textbf{Deployment Component} & \textbf{Complexity} \\
\midrule
Aggregate-authorisation preparation
& $O(|S_{c,v}|)$ \\
Aggregate AES-key encryption
& $O(1)$ per vendor; $O(V)$ overall \\
Device authentication
& $O(1)$ per device session \\
Session-binding computation
& $O(1)$ per device session \\
Bitstream variant preparation
& $O(\mu)$ \\
Provider signature generation
& $O(\mu|S_v|)$ \\
Device-side signature verification
& $O(1)$ per deployed variant \\
Aggregate-authorisation membership update
& $O(1)$ \\
\bottomrule
\end{tabular}
\end{table}

These bounds separate device-set-dependent authorization preparation from
per-session cryptographic processing and architecture-dependent bitstream
preparation. The implementation results are consistent with these
relationships: aggregate-key extraction increases with the authorized
device-set size, online AES-key encryption remains approximately
independent of that size, and membership-update cost remains approximately
constant across the evaluated cluster sizes.

\subsection{Multi-Vendor and Multi-Provider Scaling}
\label{sec:enc-scaling}

For a deployment spanning $V$ participating vendor domains, SAEID generates
one aggregate ciphertext for the AES key $m$ per vendor domain. Hence, the
online aggregate AES-key encryption workload is

\begin{equation}
O(V).
\label{eq:vendor_scalability}
\end{equation}

Once the corresponding aggregate-authorization components have been
prepared, this online cost is independent of the number of authorized
devices within each vendor domain.

For $P$ provider deployments, let $V_p$ denote the number of vendor
domains participating in provider $p$'s deployment. The cumulative
aggregate AES-key encryption workload is therefore

\begin{equation}
O\left(\sum_{p=1}^{P} V_p\right),
\label{eq:provider_scalability}
\end{equation}

which becomes $O(PV)$ when every provider deploys across all $V$ vendor
domains. This expression denotes cumulative cryptographic workload across
the provider deployments and does not represent wall-clock latency under
concurrent execution.~\autoref{tab:deployment_ops} shows the operation counts for various deployment configurations in SAEID. 

\begin{table}[t]
\centering
\caption{Operation Counts for Heterogeneous FPGA Deployments}
\label{tab:deployment_ops}
\footnotesize
\begin{tabular}{p{0.48\columnwidth}cc}
\toprule
\textbf{Deployment Configuration} &
\textbf{Bitstream Variants} &
\textbf{Aggregate AES-Key Encryption Operations} \\
\midrule
Single vendor, single architecture ($\mu=1$)
& 1 & 1 \\
Single vendor, $\mu$ architectures
& $\mu$ & 1 \\
$V$ vendors, $\mu$ architectures
& $\mu$ & $V$ \\
\bottomrule
\end{tabular}
\end{table}

The multi-vendor and multi-provider evaluation provides empirical
evidence consistent with these scaling relationships. Aggregate
encryption workload increases approximately linearly with the number of
participating vendor domains, while provider deployments maintain
separate authorization state. These results characterize cumulative
workload scaling and should not be interpreted as measurements of
end-to-end wall-clock latency for concurrent deployments.

\subsection{Communication and Deployment Scaling}

For each participating vendor, the aggregate ciphertext
$C_v=(c_1,c_2,c_3)$ has fixed size independent of $|S_{c,v}|$.
Hence, aggregate-authorization ciphertext communication is $O(V)$ for
$V$ participating vendor domains.
Authentication and session metadata contribute fixed-size overhead per
authenticated device, resulting in $O(|S_c|)$ metadata. The encrypted
bitstream payload is excluded because its size depends on the target
architecture and application.

Capability-aware clustering limits aggregate-authorization processing to
the selected devices. Adding devices within an existing vendor domain does
not increase the number of online aggregate AES-key encryption operations.
This number increases only when additional vendor domains participate.

Thus, aggregate-authorization preparation scales with the selected
devices, online aggregate AES-key encryption with the participating vendor
domains, and bitstream preparation with the number of architecture
variants. Dynamic membership updates remain $O(1)$ at the
aggregate-authorization update level.
\section{Experimental Evaluation}
\label{sec:eval}

\subsection{Experimental Setup}
\label{sec:experimental_setup}

The SAEID implementation was evaluated at both software and embedded
levels. Software experiments were conducted on an Intel Core i5-8500
3.00~GHz six-core processor running Ubuntu 22.04.2 LTS. The cryptographic
core uses PBC~0.5.14 with Type-A pairings over the 160-bit
Solinas-prime curve adopted from AgEID. AES-256-GCM was implemented
using the OpenSSL~3.0.18 EVP interface. The implementation covers the
SAEID cryptographic operations, dynamic membership, multi-vendor and
multi-provider authorization, identity-based authentication, and
identity-bound bitstream integrity.

Before performance evaluation, an automated correctness suite containing
23 functional checks was executed. The tests cover authorized and
unauthorized decryption, revoked and wrong-vendor devices, dynamic
membership, forward- and backward-membership behavior, device
authentication, and bitstream-integrity verification. All 23 tests
passed.

Performance measurements used operation-specific repetition counts.
Primitive operations were measured over 2000 repetitions, general
cryptographic and scalability experiments over 100 trials, and
decryption over 300 trials. Cross-domain isolation and membership
experiments used 200 and 100 trials, respectively. Identity-based
authentication was evaluated over 1000 sessions, and bitstream-payload
processing over 30 payload sizes. Unless otherwise stated, reported
software timings are measured execution times of the SAEID
implementation.

Published AgEID measurements are used only as reference baselines where
comparison is relevant. These include 19.4~ms per bilinear pairing,
$6.2\times10^{-3}$~ms per scalar multiplication,
$3.01\times10^{-3}$~ms per point addition, published membership-update
results, the reported 401~ms ARM Cortex-A9 decryption result, and
published FPGA resource figures. These values were not re-measured in
the present work.

For embedded evaluation, the SAEID software was cross-compiled and
executed on a physical Xilinx ZC702 platform based on the Zynq-7000
\texttt{xc7z020clg484-1} device with an ARM Cortex-A9 processor. The
evaluation executes the complete
\texttt{Setup}/\texttt{KeyGen}/\texttt{Extract}/\texttt{Encrypt}/
\texttt{Decrypt} path and evaluates recovery of the aggregate-protected
deployment secret and the corresponding AES-256 deployment key for
authorized and unauthorized devices.

A separate synthesizable RTL implementation of the elliptic-curve
scalar-multiplication component was evaluated using Vivado~2019.1 on
the same \texttt{xc7z020clg484-1} device. The RTL implements the ECC
arithmetic required for scalar multiplication but does not include a
bilinear pairing engine. Accordingly, its resource utilization and
post-implementation timing are reported as an ECC-core evaluation
rather than as a complete SAEID hardware accelerator.

Aggregate-authorization communication measurements exclude encrypted
bitstream payload transfer; payload-processing cost is evaluated
separately as a function of payload size.

\subsection{Cryptographic Performance and Scalability}
\label{sec:crypto_perf}

The cryptographic evaluation measures the principal SAEID operations and
examines the scalability of aggregate-key preparation and AES-key
encryption with the number of authorized FPGA devices. The underlying
pairing-group primitives are also measured to characterize the execution
cost of the construction.

Table~\ref{tab:crypto_perf} summarizes the measured primitive and
cryptographic-operation latencies. A bilinear pairing requires
$0.702$~ms, while scalar multiplication and point addition require
$0.954$~ms and $0.0030$~ms, respectively. Hash-to-group operations
$H_1/H_2$ require $2.188$~ms, while the SHA-256 operation used for
session-key derivation requires approximately $0.0008$~ms. At
$|S_v|=1000$, \texttt{SAEID.Extract} requires $4.477$~ms and
\texttt{SAEID.Encrypt} requires $3.729$~ms. With the
authorization-dependent decryption term precomputed,
\texttt{SAEID.Decrypt} requires $1.400$~ms.

\begin{table}[t]
\centering
\caption{Measured Cryptographic Operation Latency of SAEID}
\label{tab:crypto_perf}
\footnotesize
\begin{tabular}{lc}
\toprule
\textbf{Operation} & \textbf{Mean Time (ms)}\\
\midrule
Bilinear pairing & 0.702\\
Scalar multiplication & 0.954\\
Point addition & 0.0030\\
Hash-to-group ($H_1/H_2$) & 2.188\\
SHA-256 ($H_3$ input) & 0.0008\\
\texttt{SAEID.Extract}, $|S_v|=1000$ & 4.477\\
\texttt{SAEID.Encrypt} & 3.729\\
\texttt{SAEID.Decrypt}, precomputed term & 1.400\\
\bottomrule
\end{tabular}
\end{table}

For a vendor-specific deployment set $S_v$, \texttt{SAEID.Extract}
constructs the aggregate-authorization component by processing the
selected devices and therefore scales with $|S_v|$. Once this component
is available, \texttt{SAEID.Encrypt} generates a fixed-size aggregate
ciphertext using a constant number of cryptographic operations.

The measured extraction time increases almost linearly from 10 to
10,000 authorized devices. A linear fit to the measured data gives
\begin{equation}
T_{\mathrm{Extract}}
\approx
4.5046\times10^{-3}|S_v|-0.0151
\ \mathrm{ms},
\qquad R^2=0.9999997,
\label{eq:extract_regression}
\end{equation}
indicating that the measured scaling closely follows the expected linear
growth with $|S_v|$.

At an authorization-set size of $10{,}000$, the measured
\texttt{SAEID.Extract} time is $45.02$~ms. This scalability experiment
evaluates the software implementation with progressively larger logical
authorization sets. It does not represent a physical deployment of
$10{,}000$ FPGA devices.
In
contrast, \texttt{SAEID.Encrypt} remains approximately constant over the
same range, with measured values between $3.73$ and $3.74$~ms. The fitted
slope is approximately $-0.38$~ns per additional device and is negligible
in magnitude. Thus, the online aggregate AES-key encryption cost is
approximately independent of the number of authorized devices within a
vendor domain.

\begin{figure}[t]
\centering
\begin{tikzpicture}
\begin{axis}[
    width=\columnwidth,
    height=5.2cm,
    xlabel={Number of Authorized FPGA Devices},
    ylabel={Execution Time (ms)},
    xmode=log,
    ymode=log,
    xmin=8,
    xmax=12000,
    ymin=0.02,
    ymax=100,
    xtick={10,50,100,500,1000,5000,10000},
    xticklabels={10,50,100,500,1000,5000,10000},
    legend pos=north west,
    grid=major,
    mark size=1.8pt
]

\addplot[
    mark=*,
    thick
] coordinates {
    (10,0.028985)
    (50,0.213790)
    (100,0.439161)
    (500,2.232671)
    (1000,4.476846)
    (5000,22.526084)
    (10000,45.022735)
};
\addlegendentry{\texttt{SAEID.Extract}}

\addplot[
    mark=square*,
    thick
] coordinates {
    (10,3.739700)
    (50,3.731444)
    (100,3.741499)
    (500,3.740996)
    (1000,3.729443)
    (5000,3.736242)
    (10000,3.733051)
};
\addlegendentry{\texttt{SAEID.Encrypt}}

\end{axis}
\end{tikzpicture}
\caption{Measured scalability of SAEID aggregate-key extraction and
online AES-key encryption with increasing numbers of authorized FPGA
devices. Extraction scales linearly, whereas online encryption remains
approximately constant.}
\label{fig:enc-scale-plot}
\end{figure}
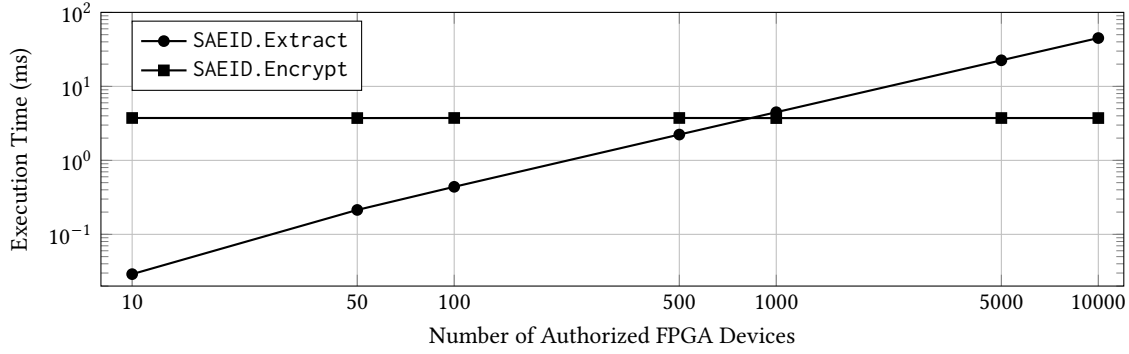

The software decryption time is $5.901$~ms when the
authorization-dependent term is computed online, reducing to
$1.400$~ms with precomputation for a 1000-device set. Thus, after
precomputation, the online decryption cost is approximately independent
of the authorized-set size. The corresponding ZC702 measurement is
reported in Section~\ref{sec:embedded_eval}. Overall, the measurements
show that the principal device-set-dependent cost lies in aggregate-key
preparation, while online encryption and precomputed-term decryption
remain approximately constant. The published AgEID pairing latency of
$19.4$~ms is used only as a reference because it was obtained on a
different platform and implementation.

\subsection{Deployment, Security, Communication, and Storage Evaluation}
\label{sec:deployment_eval}

This evaluation measures the deployment-level overhead of SAEID, including
dynamic membership, authorization metadata, cryptographic storage, device
authentication, and device-bound bitstream integrity. The analytical
scaling relationships are established in
Sections~\ref{sec:complexity}--\ref{sec:enc-scaling}; the present section
focuses on their measured overheads.

\textbf{Dynamic membership.}
Device addition and revocation modify only the affected vendor-specific
aggregate-authorization component. As shown in
Table~\ref{tab:membership}, the measured update times remain
approximately constant over the evaluated deployment sizes: device
addition requires $4.95$--$5.09~\mu$s, while revocation requires
$4.49$--$5.46~\mu$s.

\begin{table}[t]
\centering
\caption{Measured Dynamic Membership Update Overhead of SAEID}
\label{tab:membership}
\footnotesize
\begin{tabular}{lcc}
\toprule
\textbf{Operation} & \textbf{Time Range ($\mu$s)}
& \textbf{Scaling}\\
\midrule
Device addition & 4.95--5.09 & $O(1)$\\
Device revocation & 4.49--5.46 & $O(1)$\\
\bottomrule
\end{tabular}
\end{table}

The reported update cost excludes provisioning and identity-key generation
for newly enrolled devices. Unaffected devices require no credential
regeneration, and previously generated ciphertexts are not re-encrypted
solely because of a membership update.

\textbf{Communication and storage overhead.}
Each vendor-specific aggregate ciphertext contains three group elements and
occupies 384~B in the implementation. The identity-based authentication
(IBA) proof contains two group elements and adds 256~B per authenticated
device session. The following communication expressions account for the
aggregate ciphertext and IBA proof only; session-binding fields,
nonces/timestamps, and bitstream-signature metadata are excluded.

Thus, for $V$ participating vendor domains and one authenticated device
session,
\begin{equation}
C_{\mathrm{auth}}
=
(384V+256)~\mathrm{B}.
\label{eq:communication_overhead}
\end{equation}

For $N_{\mathrm{auth}}$ authenticated device sessions, the corresponding
authorization metadata is
\begin{equation}
C_{\mathrm{meta}}
=
(384V+256N_{\mathrm{auth}})~\mathrm{B}.
\label{eq:communication_sessions}
\end{equation}

\begin{table}[t]
\centering
\caption{Communication Overhead of SAEID Authorization Metadata}
\label{tab:communication}
\footnotesize
\begin{tabular}{ccc}
\toprule
\textbf{Vendor Domains} &
\textbf{Aggregate Ciphertext (B)} &
\textbf{+ IBA Proof (B)}\\
\midrule
1 & 384 & 640\\
2 & 768 & 1024\\
3 & 1152 & 1408\\
5 & 1920 & 2176\\
10 & 3840 & 4096\\
\bottomrule
\end{tabular}
\end{table}

Hence, authorization metadata grows with the number of participating
vendor domains and authenticated sessions, but remains independent of the
number of devices represented within each vendor-specific aggregate
ciphertext. The encrypted bitstream payload is excluded from
$C_{\mathrm{auth}}$ and $C_{\mathrm{meta}}$ because its size depends on
the target architecture and application. \autoref{tab:communication}
summarizes the communication overhead of SAEID authorization metadata.

At each FPGA device, SAEID stores one individual decryption key and one
identity-based private key, resulting in constant device-resident
cryptographic state with respect to deployment size, as shown in
\autoref{tab:storage}. Infrastructure-side provisioning records scale with
the number of enrolled devices, while provider identity credentials scale
with the number of registered provider--vendor domains.

\begin{table}[t]
\centering
\caption{Storage Requirements of AgEID and SAEID}
\label{tab:storage}
\footnotesize
\begin{tabular}{p{0.48\columnwidth}cc}
\toprule
\textbf{Storage Component} & \textbf{AgEID} & \textbf{SAEID}\\
\midrule
Individual decryption key & 1 group element & 1 group element\\
Identity-based device key & -- & 1 group element\\
Certificate / certificate chain & -- & --\\
Vendor provisioning records & $O(n_v)$ & $O(n_v)$\\
Provider identity key & -- & 1 per registered provider--vendor domain\\
CSP deployment metadata & $O(|S|)$ & $O(|S_c|)+O(V)$\\
Freshness state & -- & $O(1)$/device\\
\bottomrule
\end{tabular}
\end{table}

\textbf{Authentication and bitstream integrity.}
The complete software implementation of the certificate-free identity-based
authentication and identity-bound signature mechanisms was evaluated on
the Intel Core i5-8500 platform. As shown in
\autoref{tab:authoverhead}, the measured IBA device response is
$1.968$~ms, while provider-side verification requires $2.382$~ms.
Identity-bound signature generation and verification require approximately
$1.98$~ms and $2.40$~ms, respectively.

\begin{table}[t]
\centering
\caption{Measured Device Authentication and Bitstream-Verification Overhead}
\label{tab:authoverhead}
\footnotesize
\begin{tabular}{p{0.68\columnwidth}c}
\toprule
\textbf{Operation} & \textbf{Measured Time (ms)}\\
\midrule
IBA device response & 1.968\\
IBA verifier verification & 2.382\\
Identity-bound signature generation & 1.98\\
Identity-bound signature verification & 2.40\\
\bottomrule
\end{tabular}
\end{table}

Across 1000 authentication sessions, all valid device responses were
accepted. Negative tests involving incorrect device identities, modified
proof components, modified challenges, and replayed responses were
rejected. For bitstream integrity, modified or substituted payloads,
tampered AES-GCM tags, incorrect target devices, incorrect providers, and
modified release metadata were rejected in all seven tested negative
cases.

For $\mu$ architecture-specific variants, identity-bound signatures are
generated for the required variant--device associations during deployment
preparation. When signatures are generated independently for each
variant--device pair, this workload is $O(\mu|S_v|)$ for vendor domain
$v$, while verification remains $O(1)$ per deployed variant. This cost is
separate from the constant-cost authentication of an individual device
session.

\subsection{Embedded FPGA Evaluation and Resource Utilization}
\label{sec:embedded_eval}

The embedded evaluation validates the SAEID cryptographic path on the
Xilinx ZC702 platform, which uses a Zynq-7000 XC7Z020 device with an ARM
Cortex-A9 processor. The complete SAEID software path comprising
\texttt{Setup}, \texttt{KeyGen}, \texttt{Extract}, \texttt{Encrypt}, and
\texttt{Decrypt} was cross-compiled for ARM and executed on the physical
Cortex-A9 processor using GMP and PBC. Functional tests confirmed that an
authorized device successfully recovers the protected AES key $m$, whereas
an unauthorized device fails the corresponding recovery operation.

For the complete software \texttt{Decrypt()} operation, including the two
bilinear pairings, the measured execution time on the ZC702 is
$191.126$~ms, as shown in \autoref{tab:arm_decrypt}. 
The AgEID study reports $401$~ms for its Cortex-A9
decryption path. This value is retained only as a reference because the
implementation and build configurations are not identical. The present
ARM build uses GMP compiled without ARM assembly optimization, so the
measured value represents the current implementation rather than an
optimized Cortex-A9 performance limit.

\begin{table}[t]
\centering
\caption{Embedded ARM Cortex-A9 Decryption Performance}
\label{tab:arm_decrypt}
\footnotesize
\begin{tabular}{lc}
\toprule
\textbf{Platform / Implementation} & \textbf{\texttt{Decrypt()} (ms)}\\
\midrule
SAEID on ZC702 Cortex-A9 & 191.126\\
AgEID reported Cortex-A9 & 401\\
\bottomrule
\end{tabular}
\end{table}

In addition to the embedded software evaluation, a synthesizable RTL
implementation of the elliptic-curve scalar-multiplication component was
developed and evaluated using Vivado~2019.1 for the
\texttt{xc7z020clg484-1} device. The RTL evaluation focuses on this
cryptographic primitive; a complete hardware accelerator for the SAEID
protocol was not implemented. The core implements the elliptic-curve
arithmetic underlying the scalar-multiplication operations used in SAEID
key generation and aggregate-key computation. Jacobian-coordinate point
doubling, mixed point addition, modular multiplication, and modular
inversion were verified against the PBC software implementation. Eight
full 160-bit scalar-multiplication functional tests achieved bit-exact
agreement with PBC, while the modular multiplication and inversion
regression tests achieved 10/10 successful cases.

The post-implementation resource utilization is summarized in
Table~\ref{tab:resources}. The design uses a 512-bit modular-arithmetic
datapath and does not use DSP blocks. The reported power is a Vivado
vectorless estimate rather than a physical-board measurement.

\begin{table}[t]
\centering
\caption{FPGA Resource Utilization of the SAEID ECC
Scalar-Multiplication Core}
\label{tab:resources}
\footnotesize
\begin{tabular}{lc}
\toprule
\textbf{Metric} & \textbf{SAEID ECC Core}\\
\midrule
LUT utilization & 45.93\%\\
Register utilization & 17.97\%\\
Slice utilization & 70.14\%\\
DSP utilization & 0\%\\
Block RAM utilization & 0.36\%\\
Power & 0.615 W$^\dagger$\\
Timing at 100 MHz & Not met\\
Estimated frequency & $\sim$20.8 MHz$^\ddagger$\\
\bottomrule
\end{tabular}

\vspace{1mm}
\raggedright
\scriptsize
$^\dagger$Vivado vectorless estimate.\\
$^\ddagger$Estimated from the post-implementation timing result.
\end{table}

The implementation has a worst negative slack of approximately
$-38.046$~ns for the 100-MHz constraint. This timing limitation is
associated with the wide, unpipelined modular-arithmetic paths in the
current correctness-oriented design. Further pipelining and architectural
optimization could reduce the critical-path delay and improve the timing
characteristics; the resulting resource impact would depend on the
implementation architecture.

The AgEID study reports 2--17\% LUT, 1--13\% register, and 1--4\% slice
overhead, with approximately 0.1--0.2~W power. These values are retained
only as reference values because the architectural details and measurement
conditions of the AgEID hardware implementation are not sufficiently
specified for a direct resource comparison. Moreover, the present RTL
implements only the ECC scalar-multiplication component; the bilinear
pairing engine required for a complete hardware cryptographic accelerator
was not synthesized. Thus, the FPGA resource results characterize the
implemented SAEID ECC component, while the ZC702 result demonstrates the
complete software decryption path on real embedded hardware.

\section{Discussion}
\label{sec:discussion}

The results show that SAEID preserves the scalable aggregate-authorization
property of AgEID while extending it to heterogeneous multi-vendor
deployment. Aggregate-key preparation grows with the authorized set size,
whereas online aggregate AES-key encryption remains approximately
independent of the number of devices within a vendor domain. Dynamic
membership updates likewise remain effectively constant-cost, and the
encryption workload grows with the number of participating vendor domains,
consistent with the $O(V)$ construction.

The additional identity-based authentication and bitstream-integrity
mechanisms introduce bounded software overhead. The complete
authentication exchange requires approximately $4.35$~ms, while
identity-bound signature generation and verification require approximately
$1.98$~ms and $2.40$~ms, respectively. The communication overhead is also
separated by deployment factor: aggregate ciphertext metadata scales with
vendor domains, authentication metadata with authenticated sessions, and
bitstream payload size with the target architecture.

The embedded evaluation confirms that the aggregate-encryption software
path executes on the ZC702 Cortex-A9, while the RTL results validate the
implemented ECC scalar-multiplication component. The latter should be
interpreted as a component-level evaluation because a complete bilinear
pairing accelerator was not implemented. Similarly, the 10,000-device
scalability result represents a software evaluation with a large logical
authorization set rather than a physical deployment of 10,000 FPGAs.

A direct quantitative comparison with prior FPGA-security implementations
is limited by differences in cryptographic primitives, platforms, and
measurement conditions. The AgEID Cortex-A9 result is therefore retained
only as a reference baseline. The main remaining limitations are the
current ECC hardware implementation and the absence of continuous runtime
integrity and hardware-rooted attestation, which are outside the scope of
the present work.

\section{Conclusion}
\label{sec:conclusion}

This paper presented SAEID, an extension of AgEID for secure FPGA
bitstream deployment in heterogeneous multi-vendor cloud environments.
SAEID extends aggregate authorization with capability-aware deployment,
multiple IP providers, dynamic device membership, and certificate-free
identity-based device authentication and device-bound bitstream
verification, while retaining individual device recovery within each
vendor domain.
The evaluation validates the proposed framework through software,
embedded, and RTL experiments. The results confirm the expected scaling
behavior of aggregate authorization and dynamic membership, and the
complete software decryption path was validated on the ZC702 Cortex-A9
platform. The RTL evaluation further validates the ECC
scalar-multiplication component used by the cryptographic construction.
The current hardware evaluation is limited to the ECC
scalar-multiplication component and does not include a complete bilinear
pairing engine. Future work will therefore focus on a complete hardware
cryptographic accelerator, together with hardware-rooted attestation and
continuous post-configuration runtime security.

\begin{acks}
We would like to thank the Department of Science and Technology (DST),
Science and Engineering Research Board (SERB), Government of India, for
supporting this work under Grant No.~CRG/2018/2093.
\end{acks}

\appendix

\section{Correctness and Security Proof Sketches}
\label{app:proofs}

This appendix collects the correctness derivations and security proof
sketches for the SAEID framework. The underlying aggregate-encryption
construction is inherited from AgEID~\cite{debnath2024ageid}; consequently,
only the properties affected by multi-vendor authorization, dynamic
membership, identity-based authentication, and device-bound bitstream
protection are derived here. Standard properties of bilinear pairings, the
underlying cryptographic hardness assumptions, and the security of AES-GCM
are used without restating their definitions. Notation, the system and
threat models, and the protocol definitions follow the main body of the
paper (Section~\ref{sec:system-threat-model} onward).

\subsection{Correctness of Aggregate Decryption}

Let $S_{c,v}$ denote the authorized device set associated with vendor
$v$. For an authorized device $i\in S_{c,v}$, the aggregate decryption
procedure recovers the group-valued deployment secret $m_G\in G_T$ as

\begin{equation}
\widehat{m}_G
=
c_3\cdot
\frac{
\hat{e}\left(
d_{v,i}+b_{i,S_{c,v}},c_1
\right)
}{
\hat{e}\left(
g_{v,i},c_2
\right)
}.
\label{app:eq:supp-aggregate-decrypt}
\end{equation}

\begin{theorem}[Aggregate-Decryption Correctness]
For every authorized device $i\in S_{c,v}$, decryption of a ciphertext
generated for $S_{c,v}$ recovers the original group-valued deployment
secret $m_G$.
\end{theorem}

\begin{proof}[Proof sketch]
Substituting the definitions of $d_{v,i}$,
$b_{i,S_{c,v}}$, $c_1$, $c_2$, and
\[
c_3
=
m_G\cdot
\hat{e}\left(g_{v,n_v},tg_{v,1}\right)
\]
into Eq.~\eqref{app:eq:supp-aggregate-decrypt}, and applying bilinearity
of $\hat{e}(\cdot,\cdot)$, the terms introduced by the aggregate-
authorization component cancel through the same telescoping argument
as in the underlying AgEID construction. The remaining pairing factor
cancels the masking factor in $c_3$, yielding
\begin{equation}
\widehat{m}_G
=
m_G\cdot
\frac{
\hat{e}(g_{v,n_v},tg_{v,1})
}{
\hat{e}(g_{v,n_v},tg_{v,1})
}
=
m_G.
\end{equation}
Hence, every authorized device recovers the same group-valued deployment
secret $m_G$.
\end{proof}

After recovering $m_G$, the device derives the AES-256 deployment key as

\begin{equation}
\widehat{K}_{\mathrm{AES}}
=
\mathrm{SHA256}
\left(
\operatorname{Serialize}_{G_T}(\widehat{m}_G)
\right).
\label{app:eq:supp-recovered-aes-key}
\end{equation}

Since $\widehat{m}_G=m_G$ for an authorized device,

\begin{equation}
\widehat{K}_{\mathrm{AES}}
=
K_{\mathrm{AES}}.
\end{equation}

\subsection{Deployment-Secret Confidentiality}
\label{app:sec:confidentiality}

\begin{theorem}[Deployment-Secret Confidentiality]
\label{app:thm:confidentiality}
Under the BDHE assumption underlying the aggregate-encryption
construction, an adversary that does not possess the authorization
required for the target deployment cannot distinguish the protected
group-valued deployment secret $m_G\in G_T$ from a random element of
$G_T$ except with negligible probability.
\end{theorem}

\begin{proof}[Proof sketch]
SAEID preserves the aggregate-encryption algebra of the underlying
construction. In particular,
\[
c_3
=
m_G\cdot
\hat{e}(g_{v,n_v},tg_{v,1})
\]
uses the same pairing-based masking structure as the underlying
aggregate-encryption scheme, with the protected plaintext represented
as a group element in $G_T$.

A simulator can embed a BDHE challenge into the aggregate-encryption
public parameters and challenge ciphertext. Any adversary that
distinguishes the protected $m_G$ from a random element with
non-negligible advantage would therefore yield a non-negligible
advantage against the underlying BDHE-based aggregate-encryption
security game. Hence,
\begin{equation}
\operatorname{Adv}^{\mathrm{IND\text{-}CPA}}_{\mathrm{SAEID}}(\mathcal{A})
\leq
\operatorname{Adv}^{\mathrm{BDHE}}(\mathcal{B})
+
\operatorname{negl}(\lambda).
\end{equation}
Thus, $m_G$ remains confidential to unauthorized entities.
\end{proof}

The AES-256 deployment key is not exposed directly to the
aggregate-encryption layer. Instead, both provider and authorized
device derive it from the group-valued secret using
\[
K_{\mathrm{AES}}
=
\mathrm{SHA256}
\left(
\operatorname{Serialize}_{G_T}(m_G)
\right).
\]
Consequently, only an entity that successfully recovers the authorized
$m_G$ can derive the corresponding $K_{\mathrm{AES}}$. The protected
architecture-specific variants are then confidential under
AES-256-GCM, assuming unique IVs for encryptions performed under the
same key.

\subsection{Authentication Soundness and Replay Resistance}
\label{app:sec:auth-soundness}

For device identity $ID_{v,i}$, let $Q_{v,i} = H_1(ID_{v,i})$ and let
$SKID_{v,i}$ denote the corresponding identity private key. For a fresh
$(N, \tau)$, the device selects $r$ and computes
\begin{equation}
  U = r Q_{v,i}, \qquad
  h = H_2\big(ID_{v,i} \,\|\, N \,\|\, \tau \,\|\, U\big),
\end{equation}
followed by the response
\begin{equation}
  V = (r + h)\, SKID_{v,i}.
  \label{app:eq:response}
\end{equation}

The verifier checks
\begin{equation}
  \hat{e}(g, V)
  \;\overset{?}{=}\;
  \hat{e}\big(P_{\mathrm{pub},v},\, U + h Q_{v,i}\big).
  \label{app:eq:verify}
\end{equation}

The replay-resistance argument assumes that the freshness state used to
validate $N$ and $\tau$ is maintained in trusted persistent storage across
ordinary power cycles and FPGA reconfiguration. Loss or rollback of this
state is outside the present threat model.

\begin{theorem}[Authentication Soundness and
Authentication-Transcript Replay Resistance]
\label{app:thm:auth}
In the random-oracle model, an adversary without the target device's
identity private key cannot produce an accepted proof for a fresh
challenge except with negligible probability. A previously observed
authentication transcript cannot be replayed successfully in a fresh
session.
\end{theorem}

\begin{proof}[Proof sketch]
For an honestly generated response, bilinearity of the pairing gives
\begin{equation}
  \hat{e}(g, V)
  =
  \hat{e}(g, SKID_{v,i})^{\,r+h}
  =
  \hat{e}\big(P_{\mathrm{pub},v},\, U + h Q_{v,i}\big),
\end{equation}
so the verification equation in Eq.~\eqref{app:eq:verify} holds.

For an adversary, producing a valid response for a fresh Fiat--Shamir
challenge without the corresponding identity key contradicts the
soundness of the underlying identity-based proof in the random-oracle
model, under the stated hardness assumption. Moreover, the challenge
contains the fresh values $N$ and $\tau$; therefore, a response from an
earlier session is bound to a different challenge and cannot be directly
replayed.
\end{proof}

\subsection{Authenticated Session Binding}
\label{app:sec:session-binding}

Following successful device authentication, the provider and device
derive the shared pairing value
\begin{equation}
K_{v,i,p}
=
\hat{e}
\left(
SK_{ID_{v,i}},Q_p^{(v)}
\right)
=
\hat{e}
\left(
Q_{v,i},SK_{ID_p}^{(v)}
\right).
\label{app:eq:supp_shared_pairing}
\end{equation}
The authenticated session key is then derived as
\begin{equation}
K_{\mathrm{sess},v,i}
=
H_3
\left(
K_{v,i,p}
\parallel
ID_{v,i}
\parallel
ID_p^{(v)}
\parallel
N_{v,i}
\parallel
\tau_{v,i}
\right).
\label{app:eq:supp_session_key}
\end{equation}

\begin{theorem}[Authenticated Session Binding]
\label{app:thm:session-binding}
Assuming the security of $H_3$ and the underlying MAC, an adversary that
does not possess the authenticated session key
$K_{\mathrm{sess},v,i}$ cannot produce a valid session-binding tag for a
modified, redirected, or cross-session aggregate ciphertext except with
negligible probability.
\end{theorem}

\begin{proof}[Proof sketch]
The session key is derived from the authenticated shared pairing value
$K_{v,i,p}$, the identities of the target device and provider, and the
fresh session values $N_{v,i}$ and $\tau_{v,i}$. Thus, under the
authentication and freshness guarantees of Theorem~\ref{app:thm:auth}, an adversary
cannot obtain or reuse the session key associated with a different
authenticated session or identity context except with negligible
probability.

The tag input $D_{v,i}$ additionally includes the provider identity,
target-device identity, vendor domain, capability class, fresh session
values, and aggregate ciphertext. Consequently, modifying any of these
fields changes the MAC input. Producing a valid tag for the modified
context, or reusing a tag from another session or deployment context,
would therefore constitute a successful forgery against the underlying
MAC. Hence, a modified, redirected, or cross-session aggregate ciphertext
is rejected except with negligible probability.
\end{proof}

\subsection{Bitstream Integrity and Device Binding}

For each protected architecture-specific variant
$\mathrm{Ct}^{(k)}$ intended for $\mathrm{FPGA}_{v,i}$, the registered
provider generates an identity-bound proof using its vendor-scoped identity
private key $\mathrm{SK}_{\mathrm{ID}_p}^{(v)}$. The provider first samples
fresh randomness
$r' \xleftarrow{R}\mathbb{Z}_p$ and computes

\begin{equation}
U_{\sigma}^{(k)}
=
r' Q_p^{(v)},
\label{app:eq:supp-sig-commitment}
\end{equation}

where
$Q_p^{(v)}=H_1(ID_p^{(v)})$ is the provider's identity point.

The Fiat--Shamir challenge binds the provider identity, target-device
identity, encrypted payload, release metadata, and the proof commitment:

\begin{equation}
h_{\sigma}^{(k)}
=
H_2
\left(
ID_p^{(v)}
\parallel
ID_{v,i}
\parallel
\mathrm{Ct}^{(k)}
\parallel
\mathrm{Hdr}^{(k)}
\parallel
U_{\sigma}^{(k)}
\right).
\label{app:eq:supp-sig-challenge}
\end{equation}

The provider computes the response

\begin{equation}
V_{\sigma}^{(k)}
=
\left(
r' + h_{\sigma}^{(k)}
\right)
\mathrm{SK}_{\mathrm{ID}_p}^{(v)}.
\label{app:eq:supp-sig-response}
\end{equation}

The resulting identity-bound signature is

\begin{equation}
\sigma_{p,v,i}^{(k)}
=
\left(
U_{\sigma}^{(k)},
V_{\sigma}^{(k)}
\right).
\label{app:eq:supp-sig}
\end{equation}

At the target FPGA, the device recomputes
$Q_p^{(v)}=H_1(ID_p^{(v)})$ and
$h_{\sigma}^{(k)}$ and accepts the signature only if

\begin{equation}
\hat e
\left(
g,V_{\sigma}^{(k)}
\right)
\stackrel{?}{=}
\hat e
\left(
P_{\mathrm{pub},v},
U_{\sigma}^{(k)}
+
h_{\sigma}^{(k)}Q_p^{(v)}
\right).
\label{app:eq:supp-sig-verify}
\end{equation}

For an honestly generated signature,
$\mathrm{SK}_{\mathrm{ID}_p}^{(v)}=s_vQ_p^{(v)}$ and
$P_{\mathrm{pub},v}=s_vg$. Therefore, by bilinearity,

\begin{align}
\hat e(g,V_{\sigma}^{(k)})
&=
\hat e
\left(
g,
(r'+h_{\sigma}^{(k)})\mathrm{SK}_{\mathrm{ID}_p}^{(v)}
\right)
\\
&=
\hat e
\left(
P_{\mathrm{pub},v},
(r'+h_{\sigma}^{(k)})Q_p^{(v)}
\right)
\\
&=
\hat e
\left(
P_{\mathrm{pub},v},
U_{\sigma}^{(k)}
+
h_{\sigma}^{(k)}Q_p^{(v)}
\right),
\end{align}

since
$U_{\sigma}^{(k)}=r'Q_p^{(v)}$. Hence, an honestly generated
signature satisfies Eq.~\eqref{app:eq:supp-sig-verify}.

\begin{theorem}[Bitstream Integrity and Device Binding]
\label{app:thm:integrity}
Assuming unforgeability of the identity-bound provider proof and
authenticity of AES-GCM, an adversary cannot modify, substitute, or
redirect a protected bitstream to another target device while causing it
to be accepted, except with negligible probability.
\end{theorem}

\begin{proof}[Proof sketch]
The challenge in Eq.~\eqref{app:eq:supp-sig-challenge} cryptographically binds
the provider identity, target-device identity, encrypted payload,
release metadata, and proof commitment. Therefore, redirecting a valid
signature to another target device changes the challenge through
$ID_{v,i}$ and invalidates the verification equation. Similarly,
modifying or substituting $\mathrm{Ct}^{(k)}$ or
$\mathrm{Hdr}^{(k)}$ changes $h_{\sigma}^{(k)}$ and requires the adversary
to produce a valid response without the corresponding provider identity
private key, contradicting the assumed unforgeability of the
identity-bound provider proof.

An adversary that attempts to modify the encrypted bitstream payload
without generating a corresponding valid provider proof is likewise
rejected by signature verification. Even if a modified ciphertext is
presented with a syntactically valid signature component, the
AES-GCM authentication tag detects any unauthorized modification of the
encrypted payload, except with negligible probability under the
authenticity of AES-GCM. Hence, a protected bitstream is accepted only
when the provider identity, intended target device, encrypted payload,
and release metadata jointly satisfy the identity-bound verification and
authenticated-encryption checks. These checks are performed before
bitstream decryption and FPGA configuration.
\end{proof}

\subsection{Dynamic-Membership Security}
\label{app:sec:membership}

For a vendor-specific authorized set $S_{c,v}$, membership changes update
only the corresponding aggregate-authorization component:
\begin{equation}
K_{S_{c,v}\cup\{i'\},\,v}
=
K_{S_{c,v},\,v}
+
\Delta^{\mathrm{add}}_{v,i'},
\qquad
K_{S_{c,v}\setminus\{i^*\},\,v}
=
K_{S_{c,v},\,v}
-
\Delta^{\mathrm{rev}}_{v,i^*}.
\label{app:eq:membership-update}
\end{equation}

Let $K_{S_{c,v},v}^{(e)}$ denote the aggregate-authorization state
associated with membership epoch $e$. A membership change creates a new
epoch, and subsequently generated ciphertexts use the updated
authorization state. Previously generated ciphertexts remain associated
with their original membership epoch.

\begin{theorem}[Forward and Backward Membership Security]
\label{app:thm:membership}
A revoked device cannot decrypt ciphertexts generated after its revocation,
while a newly enrolled device cannot decrypt ciphertexts generated before
its enrollment, except with negligible probability under the underlying
aggregate-encryption security assumption.
\end{theorem}

\begin{proof}[Proof sketch]
Revocation removes the corresponding device contribution from the
aggregate-authorization structure used for subsequent membership epochs
(Eq.~\eqref{app:eq:membership-update}). Thus, a revoked device lacks the
authorization relation required to recover the group-valued deployment
secret $m_G$ protected under a subsequently generated ciphertext.
Conversely, a newly enrolled device is absent from the authorization
structures associated with earlier epochs and therefore cannot use its
newly issued credentials to recover the deployment secrets protected under
earlier ciphertexts.

The update itself requires only a constant number of group operations.
Provisioning and identity-key generation for a newly enrolled device are
separate enrollment operations. These guarantees apply to ciphertexts
generated under the corresponding membership epochs and do not revoke
plaintext, deployment secrets, or bitstreams that an authorized device
legitimately obtained before its revocation.
\end{proof}

\subsection{Multi-Vendor, Multi-Provider, and CSP Isolation}
\label{app:sec:isolation}

\begin{theorem}[Domain Isolation and CSP Security]
\label{app:thm:isolation}
Under the SAEID key-separation model, compromise of one vendor or provider
domain does not directly provide the cryptographic capability to decrypt or
forge information belonging to another independent domain. A CSP-only
adversary likewise cannot recover protected AES keys, forge an authenticated
deployment session, or cause a modified bitstream to be accepted, except
with negligible probability.
\end{theorem}

\begin{proof}[Proof sketch]
Each vendor maintains separate aggregate-encryption and identity
cryptographic domains, based on independent secrets $\gamma_v$ and $s_v$,
respectively. Therefore, authorization material from one vendor domain is
not interchangeable with that of another vendor. Provider credentials are
similarly scoped to their registered identity domains and cannot be used
to authenticate another provider.

A CSP observes ciphertexts, deployment metadata, and protocol messages but
does not possess vendor master secrets, identity-authority secrets, or
device private keys. AES-key confidentiality follows from
Theorem~\ref{app:thm:confidentiality}, device authentication and
authentication-transcript replay resistance from
Theorem~\ref{app:thm:auth}, authenticated session binding from
Theorem~\ref{app:thm:session-binding}, and bitstream integrity and device
binding from Theorem~\ref{app:thm:integrity}. Hence, successful cross-domain
decryption, device or provider impersonation, session-binding forgery, or
unauthorized bitstream modification would require compromise of protected
domain secrets or a break of the corresponding cryptographic primitive.
\end{proof}

\subsection{Composition of Security Properties}
\label{app:sec:composition}

The SAEID security properties rely on distinct cryptographic mechanisms and
their corresponding security assumptions. Aggregate encryption protects
the deployment AES key $m$ under the vendor-specific aggregate-encryption
secret $\gamma_v$, while device authentication and identity-bound provider
signatures use the vendor identity secret $s_v$. The session-binding key
$K_{\mathrm{sess},v,i}$ is derived from the authenticated device--provider
exchange and is used only to authenticate the deployment context. The
architecture-specific bitstream variants are subsequently protected using
AES-256-GCM under $m$.

These mechanisms provide complementary security guarantees. The
aggregate-encryption layer establishes confidentiality of $m$ and
authorization for the selected device set. The identity-based
authentication layer establishes the identity of the participating device,
while the identity-bound provider proof binds the protected bitstream to
the registered provider and intended target device. The session-binding
MAC additionally binds the aggregate ciphertext to the authenticated
deployment context and fresh session values. AES-GCM provides
authenticated protection of the encrypted bitstream payload.

Under the stated key-separation, freshness, aggregate-encryption, identity
proof, MAC, and AES-GCM security assumptions, an end-to-end security
violation must result from the failure of at least one of the corresponding
underlying mechanisms. Thus, the overall failure probability can be
bounded by the union of the failure probabilities of the individual
security mechanisms. The composition does not require the CSP to be
trusted, since the CSP does not possess the vendor master secrets,
identity private keys, provider credentials, session keys, or deployment
AES key.

The individual security arguments established in this appendix therefore
cover the principal security properties of the SAEID deployment workflow:
AES-key confidentiality, device authentication and replay resistance,
authenticated session binding, bitstream integrity and device binding,
dynamic-membership security, and multi-vendor, multi-provider, and CSP
isolation.

\subsection{Summary of Security Guarantees}
\label{app:sec:summary}

Table~\ref{app:tab:summary} summarizes the principal security properties
established by the arguments above.

\begin{table}[h]
\caption{Summary of SAEID security guarantees}
\label{app:tab:summary}
\small
\begin{tabularx}{\linewidth}{@{}lX@{}}
\toprule
\textbf{Property} & \textbf{Basis} \\
\midrule
AES-key confidentiality
& Underlying BDHE-based aggregate-encryption security \\

Authentication soundness
& Identity-based proof soundness and Fiat--Shamir security \\

Authentication-transcript replay resistance
& Fresh nonce and timestamp in the authentication challenge \\

Authenticated session binding
& Session-key derivation, fresh session context, and MAC security \\

Bitstream integrity
& Identity-bound provider proof unforgeability and AES-GCM authenticity \\

Device binding
& Target-device identity included in the identity-bound proof challenge \\

Dynamic-membership security
& Updated aggregate authorization and aggregate-encryption security \\

Multi-vendor/provider isolation
& Separate vendor aggregate-encryption and identity cryptographic domains \\

CSP resistance
& Composition of confidentiality, authentication, session binding, and integrity \\
\bottomrule
\end{tabularx}
\end{table}

\end{document}